\documentclass[envcountreset]{llncs}
\usepackage[ruled, lined,vlined,linesnumbered,commentsnumbered]{algorithm2e}
\usepackage{amsmath}
\usepackage{amssymb}
\usepackage{graphicx}
\usepackage{setspace}
\usepackage{bm}
\usepackage{algpseudocode}
\usepackage{appendix}
\usepackage{color}

\title{Solving Multi-choice Secretary Problem in Parallel: An Optimal Observation-Selection Protocol\thanks{The work is partially supported by National Natural Science Foundation of China (61170062, 61222202, 61433014, 61173009).}}
\author{Xiaoming Sun\and Jia Zhang\and Jialin Zhang}
\institute{Institute of Computing Technology,
           Chinese Academy of Sciences, Beijing, China}
\date{}
\def\gto{\gamma_{2,1}}
\def\gtt{\gamma_{2,2}}
\def\got{\gamma_{1,2}}
\def\goo{\gamma_{1,1}}
\def\ittt{\iqjk{2}{2}{2}}
\def\itto{\iqjk{2}{2}{1}}
\def\itot{\iqjk{2}{1}{2}}
\def\itoo{\iqjk{2}{1}{1}}

\def\iooo{\iqjk{1}{1}{1}}
\newcommand{\iqjk}[3]{i_{#1,#2,#3}}
\newcommand{\yqijk}[4]{y_{#1,#2}^{#3|#4}}

\newcommand{\ttoi}[1]{\floor{\frac{#1-1}{Q}}+1}
\newcommand{\swjkq}{shared $(Q,J,K)$ problem}
\newcommand{\ewjkq}{exclusive $(Q,J,K)$ problem}
\newcommand{\fswjkq}{shared $Q$-queue $J$-choice $K$-best secretary problem}
\newcommand{\fewjkq}{exclusive $Q$-queue $J$-choice $K$-best secretary problem}

\newcommand{\floor}[1]{\left\lfloor #1\right\rfloor}

\vspace{0.2cm}

\newtheorem{defi}{Definition}

\newtheorem{lemm}{Lemma}
\newtheorem{prop}{Proposition}
\newtheorem{theo}{Theorem}
\newtheorem{fact}{Fact}

\algrenewcommand{\Require}{\State \textbf{Input: }}
\algrenewcommand{\Ensure}{\State \textbf{Output: }}
\begin{document}
\maketitle
\bibliographystyle{plain}

\vspace{-0.2cm}
\begin{abstract}
	The classical secretary problem investigates the question of how to hire the best secretary from $n$ candidates who come in a uniformly random order.
	In this work we investigate a parallel generalizations of this problem introduced by Feldman and Tennenholtz \cite{Feldman:2012}. We call it \emph{\fswjkq}. In this problem, $n$ candidates are evenly distributed into $Q$ queues, and instead of hiring the best one, the employer wants to hire $J$ candidates among the best $K$ persons. The $J$ quotas are \emph{shared} by all queues. This problem is a generalized version of $J$-choice $K$-best problem which has been extensively studied and it has more practical value as it characterizes the parallel situation.
	
	Although a few of works have been done about this generalization, to the best of our knowledge, no optimal deterministic protocol was known with general $Q$ queues. In this paper, we provide an optimal deterministic protocol for this problem. The protocol is in the same style of the $1\over e$-solution for the classical secretary problem, but with multiple phases and adaptive criteria. Our protocol is very simple and efficient, and we show that several generalizations, such as the \emph{fractional $J$-choice $K$-best secretary problem} and \emph{\fewjkq}, can be solved optimally by this protocol with slight modification and the latter one solves an open problem of Feldman and Tennenholtz \cite{Feldman:2012}.
	In addition, we provide theoretical analysis for two typical cases, including the $1$-queue $1$-choice $K$-best problem and the shared $2$-queue $2$-choice $2$-best problem. For the former, we prove a lower bound $1-O(\frac{\ln^2K}{K^2})$ of the competitive ratio. For the latter, we show the optimal competitive ratio is $\approx0.372$ while previously the best known result is 0.356 \cite{Feldman:2012}.
	\vspace{-0.3cm}
\end{abstract}
\thispagestyle{plain}

\section{Introduction}
\vspace{-0.2cm}
The classical \emph{secretary problem} considers the situation that an employer wants to hire the best secretary from $n$ candidates that come one by one in a uniformly random order~\cite{Gardner:1960}. Immediately after interviewing a candidate, the employer has to make an irrevocable decision of whether accepting this candidate or not. The goal of the employer is to maximize the probability of hiring the best one among these candidates.
It is well known that the optimal solution is in a phase style: the employer firstly interviews $n/e$ candidates without selecting anyone, then, he/she chooses the first candidate who is better than all previous ones. This protocol hires the best candidate with probability $1/e$ and it is optimal~\cite{Dynkin:1963,Lindley:1961}. This problem captures many scenarios 

For example, the employer may hire the candidate before a more suitable interviewee arrives, the seller may sell the item without knowing the future buyer who offers higher price, the company may accept an order without the knowledge that the future task is more urgent.
This problem has been studied extensively in many fields, such as decision theory~\cite{Lindley:1961}, game theory~\cite{Babaioff:2008,Hajiaghayi:2004,Kleinberg:2005} and theory of computation~\cite{Borosan:2009,Freeman:1983}, etc. 

The classical secretary problem has many generalizations. A natural kind of generalizations is to relax the requirement that only selecting the best one and instead, allow the employer hiring multiple candidates.
	Kleinberg~\cite{Kleinberg:2005} considered that the employer selects multiple candidates with the objective to maximize the expectation of the total {\em values} of selected persons, and he proposed the first protocol whose expected competitive ratio tends to $1$ when the number of choices goes to infinity.
	Buchbinder et al.~\cite{Buchbinder:2010} revealed an important relationship between the secretary problem and linear programming, which turns out to be a powerful method to construct optimal (randomized) protocols for many variants of secretary problems. Those variants include the so called $J$-choice $K$-best problem that the employer wants to hire $J$ candidates from the best $K$ candidates of all. For the $J$-choice $K$-best problem, they construct a randomized optimal algorithm based on the optimal solution of corresponding linear program.
	Another important variant is proposed by Feldman et al.~\cite{Feldman:2012}.
	They were the first to introduce the parallel model. In their work, the candidates are divided into several queues to be interviewed by different interviewers. They studied two interesting settings: the quotas are pre-allocated and the quotas are shared by all interviewers. For these settings, they designed algorithms and analyzed the competitive ratios based on the random time arrival model \cite{Feldman:2011}. 
	Chan et al. \cite{Chan:2013} combined the results of Buchbinder et al. \cite{Buchbinder:2010} with the random time arrival model \cite{Feldman:2011} and considered infinite candidates.
Under their model, they constructed a $(J,K)$-threshold algorithm for $J$-choice $K$-best problem. They also showed that their infinite model can be used to capture the asymptotic behavior of the finite model.


In this work, we focus on the shared parallel model introduced by Feldman et al. \cite{Feldman:2012}. All the algorithms and analysis are based on the classical discrete and finite model. The parallel model can characterize many important situations where resource is limited or low latency is required.
A typical case is the emergency diagnosis in hospital. To shorten the waiting time, patients are diagnosed by ordinary doctors in parallel. 
The serious patients are selected to be diagnosed by the expert doctors, since the experts are not enough and they can only deal with limited number of patients.

Our main result is an optimal deterministic protocol, which we call \emph{Adaptive Observation-Selection Protocol}, for the \textit{\fswjkq} (abbreviated as \emph{\swjkq}). In this problem, $n$ candidates are assigned to $Q$ queues and interviewed in parallel. All queues \emph{share} the $J$ quotas. Besides, there is a set of weights $\{w_k\,|\,1\leq k\leq K\}$ where $w_k$ stands for how important the $k$-th rank is. The employer wants to maximize the expectation of the summation of the weight associated with the selected secretaries.
To design an optimal protocol, we generalize the linear program technique introduced by Buchbinder et al. \cite{Buchbinder:2010}. 
To design an optimal protocol, we generalize the linear program technique introduced by Buchbinder et al. \cite{Buchbinder:2010}. 
Based on the optimal solution of LP model, one can design a randomized optimal algorithm. However, it is time consuming to solve the LP (the LP has $nJK$ variables) and the randomized algorithm is unpractical to apply. Besides, although this LP model has been adopted in many work, its structure hasn't been well studied in general. With digging into its structure, we develop a nearly linear time algorithm to solve the LP within $O(nJK^2)$ time. More importantly, our protocol is deterministic. It is also simple and efficient. After 
We show that this is not the case by providing a simple deterministic counterpart for \swjkq. The key observation we use is that, besides the close relationship between the protocol of secretary problem and the feasible solution of linear program, the structure of the optimal solution reveals the essences of such problem, and actually points out the way to design a clean and simple deterministic protocol.
Our protocol can be extended to solve other extensions, as their LP models have the similar structure essentially. Among those extensions, the optimal protocol for \emph{\fewjkq} addresses an open problem in the work of Feldman et al. \cite{Feldman:2012}.

Our protocol is a nature extension of the well known $1/e$-protocol of the classical problem.
In the $1/e$-protocol, the employer can treat the first $n/e$ candidates as an {\em observation phase} and set the best candidate in this phase to be a criteria. In the second phase, the employer makes decision based on this criteria. In our problem, it is natural to extend the above idea to multiple phases in each queue and the criteria may change in different phases.
Actually, the similar intuition has been used in many previous works, not only the secretary problem \cite{Babaioff:2007,Feldman:2012}, but also some other online problems such as online auction \cite{Hajiaghayi:2004} and online matching \cite{Kesselheim:2013}. This intuition seems straightforward, but it is hard to explain why it works.
In this work, we theoretically prove that this intuition indicates the right way and can lead to optimality in our case.

 	Another contribution is that we provide theoretical analysis for the competitive ratio of non-weighted cases of our problem.
	For the \emph{ $(1,1,K)$ case}, we provide a lower bound $1-O\left(\frac{\ln^2K}{K^2}\right)$ and some numerical results. For the \emph{shared $(2,2,2)$ case}, we show that the optimal competitive ratio is approximately 0.372 which is better than 0.356 that obtained by Feldman et al. \cite{Feldman:2012}.




\vspace{0.15cm}
\noindent{\bf More Related Work} ~Besides those results mentioned above, there are lots of works that are closely related to this one.
Ajtai et al. \cite{Ajtai:2001} have considered the $K$-best problem with the goal to  minimize the expectation of the sum of the ranks (or powers of ranks) of the accepted objects.
In the \textit{Matroid secretary problem}~\cite{Babaioff:2007:Mar,Chakraborty:2012,Dimitrov:2008,Dinitz:2013,Gharan:2013,Jaillet:2013,Im:2011,Korula:2009,Soto:2013}, it introduces some combinatorial restrictions (called matroid restriction) to limit the possible set of selected secretaries. Another kind of combinatorial restriction is the knapsack constraints~\cite{Babaioff:2007,Babaioff:2008}. They combined the online knapsack problem and the idea of random order in secretary problem. Another branch of works consider the value of selected secretaries. It is no longer the summation of values of each selected one, but will be a submodular function among them~\cite{Bateni:2010,Feldman:2011,Gupta:2010}. Besides, Feldman et al. \cite{Feldman:2011}
considered the secretary problem from another interesting view. They assumed all of the candidates come to the interview at a random time instead of a random order. Some works talked about the case that only partial order between candidates are known for the employer~\cite{Georgiou:2008,Kumar:2011}.
There are also some works considering the secretary problem from the view of online auction~\cite{Babaioff:2007,Babaioff:2008,Babaioff:2007:Mar,Hajiaghayi:2004,Kesselheim:2013,Kleinberg:2005,Koutsoupias:2013,Lavi:2000}. In these works, one seller wants to sell one of more identical items to $n$ buyers, and the buyers will come to the market at different time and may leave after sometime.
The goal of the seller is to maximize his/her expected revenue as well as the concern of truthfulness.

\section{Preliminaries \label{sec:name:pre}}
In this section we formally define the \swjkq.
Given positive integers $Q,\,J,\,K$ and $n$ with $Q,\,J,\,K\leq n$, suppose the employers want to hire $J$ secretaries from $n$ candidates that come one by one in a uniformly random order. There are $Q$ interviewers. Due to practical reason, like time limitation, they do the interview in parallel.
All candidates are divided into $Q$ queues, that is, the $i$-th person is assigned to the queue numbered $i\!\!\!\mod Q$ ($i=1,\ldots,n$). The employers then interview those candidates simultaneously. All the $J$ quotas are shared by the $Q$ queues. That means in each queue, the employers can hire a candidate if the total number of hired persons is less than $J$. The only information shared among $Q$ queues is the number of the candidates already hired. Thus the employer in each queue only knows the relative order about those candidates already interviewed in his/her own queue but has no idea about those unseen ones and persons in other queues. After interviewing each candidate, the employer should make an irrevocable decision about whether employ this
candidate or not. For the sake of fairness, we make a reasonable assumption that the duration of the interviewing for each candidate is uniform and fixed. This ensures the interview in each queue is carried out in the same pace.
When employers in several queues want to hire the candidate in their own queues at the same time, to break the tie, the queues with smaller number have higher priority.
Besides, we suppose the employers only value the best $K$ candidates and assign different weights to  every one of the $K$ candidates and those weights satisfies $w_1\geq w_2\geq \cdots\geq w_K > 0$ where the $w_k$ stands for the importance of the $k$-th best candidate in the employer's view. Candidates not in best $K$ can be considered have a weight 0. The object function is to maximize the expectation of the summation of the weight of selected candidates. This is the so called \emph{\fswjkq}, and we abbreviate it as \emph{\swjkq} for convenience.

\vspace{-0.2cm}
\section{Optimal Protocol for Shared $(Q,J,K)$ Problem \label{sec:name:qjk}}
\vspace{-0.1cm}
In this section, we first characterize the \swjkq\ by a linear program and then construct a deterministic protocol for the \swjkq. We will talk about the relationship between the linear program and our protocol, and finally use the idea of primal and dual to show our protocol is optimal. 
\vspace{-0.2cm}
\subsection{Linear Program for the Shared $(Q,J,K)$ Secretary Problem\label{sec:name:lpmodel}}
\vspace{-0.1cm}
We use a linear program to characterize the \swjkq\  and provide its dual program. This approach was introduced by Buchbinder et al. \cite{Buchbinder:2010} to model the $J$-choice $K$-best problem. We are the first to generalize it to the \swjkq.

\vspace{0.1cm}
\noindent{\bf Primal Program for the Shared $(Q,J,K)$ Problem}\ 
Without loss of generality, we assume $n$ is a multiple of $Q$. Let $c_{q,i}$ stand for the $i$-th candidate in $q$-th queue and $x_{q,i}^{j|k}$ stand for the probability that $c_{q,i}$ is selected as the $j$-th one given that he/she is the $k$-th best person up to now in $q$-th queue. When the $J,\,K$ and the weights are given, we know the offline optimal solution is $\sum_{l=1}^{\min(J,K)}w_l$. We denote it as $W$. Then we can model the \swjkq\ as follow. 
{\small
\begin{equation}
\label{sec3:equ:lpprimal}
    \begin{split}
		&\qquad\max z=\frac{1}{nW} \sum_{q=1}^Q\sum_{j=1}^J\sum_{l=1}^K\sum_{i=1}^n\sum_{k=l}^Kw_k \frac{\binom{i-1}{l-1}\binom{n-i}{k-l}}{\binom{n-1}{k-1}}x_{q,i}^{j|k}\\
        &\text{s.t.}\left\{
        \begin{aligned}
			&x_{q,i}^{j|k}\leq \sum_{m=1}^Q\sum_{s=1}^{i-1}\frac{1}{s}\sum_{l=1}^K\left(x_{m,s}^{j-1|l}-x_{m,s}^{j|l}\right)+\sum_{m=1}^{q-1}\frac{1}{i}\sum_{l=1}^K\left(x_{m,i}^{j-1|l}-x_{m,i}^{j|l}\right),& \\
			&(1\leq q\leq Q,\,1\leq i\leq n/Q,\,1\leq k\leq K,\,1\leq j\leq J)& \\
   &x_{q,i}^{j|k}\geq 0, ~~ (1\leq q\leq Q,\,1\leq i\leq n/Q,\,1\leq k\leq K,\,1\leq j\leq J). &
        \end{aligned}
        \right.\\
    \end{split}
\end{equation}
}

We briefly explain this program. As we can see, $c_{q,i}$ will be selected in $j$-th round only if there are exact $j-1$ candidates are selected \emph{before} $c_{q,i}$. Consequently, according to the definition of $x_{q,i}^{j|k}$, it is clear that $x_{q,i}^{j|k}$ must be less than the probability that $j-1$ candidates are selected. Thus we have the following inequality.

\begin{equation*}
	\begin{split}
		x_{q,i}^{j|k} & \leq  \Pr(\text{there are at least $j-1$ candidates are selected before $c_{q,i}$})\\
					  & \qquad-\Pr(\text{there are at least $j$ candidates are selected before $c_{q,i}$})\\
				   & = \sum_{s = 1}^{i-1}\sum_{m = 1}^{Q}\left(\Pr(\text{$c_{m,s}$ is selected in $(j-1)$-th round})-\Pr(\text{$c_{m,s}$ is selected in $j$-th round})\right)\\
				   &\qquad + \sum_{m = 1}^{q-1}\left(\Pr(\text{$c_{m,i}$ is selected in $(j-1)$-th round})-\Pr(\text{$c_{m,i}$ is selected in $j$-th round})\right)\\
				   &= \sum_{m=1}^Q\sum_{s=1}^{i-1}\frac{1}{s}\sum_{l=1}^K\left(x_{m,s}^{j-1|l}-x_{m,s}^{j|l}\right)+\sum_{m=1}^{q-1}\frac{1}{i}\sum_{l=1}^K\left(x_{m,i}^{j-1|l}-x_{m,i}^{j|l}\right)\\
	\end{split}
\end{equation*}

Note that when $j=1$, the constraint actually is $$x_{q,i}^{1|k}\leq 1-\sum_{m=1}^Q\sum_{s=1}^{i-1}\frac{1}{s}\sum_{l=1}^Kx_{m,s}^{1|l}-\sum_{m=1}^{q-1}\frac{1}{i}\sum_{l=1}^Kx_{m,i}^{1|l}.$$ However, for the convenience of description, we add a set of dummy variables $x_{q,i}^{0|k}$, and set $x_{1,1}^{0|1} = 1$ while others to be 0. This makes $$\sum_{m=1}^Q\sum_{s=1}^{i-1}\frac{1}{s}\sum_{l=1}^Kx_{m,s}^{0|l}+\sum_{m=1}^{q-1}\frac{1}{i}\sum_{l=1}^Kx_{m,i}^{j|l} = 1,$$ so that the LP has a uniform constraint for all $1\leq j\leq J$. 
 
Consider the object function. 
Let $X$ stand for the random variable of the summation of weights of the selected candidates. Then, we have
\begin{equation*}
	\begin{split}
		E(X) &= \sum_{q = 1}^Q\sum_{i=1}^{n/Q}\sum_{j=1}^J\sum_{l=1}^K x_{q,i}^{j|l}\sum_{k = l}^K \Pr(\text{$c_{q,i}$ is $k$-th best candidate})\cdot w_k\\
		&= \sum_{q = 1}^Q\sum_{i=1}^{n/Q}\sum_{j=1}^J\sum_{l=1}^K x_{q,i}^{j|l}\sum_{k = l}^K \frac{\binom{i-1}{l-1}\binom{n-i}{k-l}}{n\binom{n-1}{k-1}}w_k.\\
	\end{split}
\end{equation*}
Thus, the competitive ratio is $\frac{E(X)}{W}$. It is just our objective function.

For further analysis, we provide several definitions about the primal program.

\begin{defi}[Crucial Constraint]
    We call the constraint 
	{\small
    $$x_{q,i}^{j|k}\leq \sum_{m=1}^Q\sum_{s=1}^{i-1}\frac{1}{s}\sum_{l=1}^K\left(x_{m,s}^{j-1|l}-x_{m,s}^{j|l}\right)+\sum_{m=1}^{q-1}\frac{1}{i}\sum_{l=1}^K\left(x_{m,i}^{j-1|l}-x_{m,i}^{j|l}\right)$$ 
	}
for $1\leq q\leq Q,\, 1\leq i\leq n/Q,\, 1\leq j\leq J,\,1\leq k\leq K$, the \emph{crucial constraint} for $x_{q,i}^{j|k}$.
    
\end{defi}

\begin{defi}[$(0,1)$-solution and Crucial Position]
     Given a feasible solution of the primal program, if there are $JKQ$ points $\{i'_{q,j,k}\,|\,1\leq q\leq Q,\ 1\leq j\leq J,\ 1\leq k\leq K\}$ satisfy 
	{\small
    \begin{equation*}
        x_{q,i}^{j|k}=\left\{
        \begin{aligned}
            & \sum_{m=1}^Q\sum_{s=1}^{i-1}\frac{1}{s}\sum_{l=1}^K\left(x_{m,s}^{j-1|l}-x_{m,s}^{j|l}\right)+\sum_{m=1}^{q-1}\frac{1}{i}\sum_{l=1}^K\left(x_{m,i}^{j-1|l}-x_{m,i}^{j|l}\right)>0,&
            &i'_{q,j,k}\leq i\leq n/Q&\\
            &  0,&
            &1\leq i < i'_{q,j,k}&
        \end{aligned}
        \right.
    \end{equation*}
     }
    for all $1\leq q\leq Q$, $1\leq j\leq J$, $1\leq k\leq K$, we call this feasible solution \emph{$(0,1)$-solution} of the primal program, and $i'_{q,j,k}$ is the crucial position for $x_{q,i}^{j,k}$.
\end{defi}

Note that, in a $(0,1)$-solution, only when $x_{q,i}^{j|k} > 0$, we consider the crucial constraint for the $x_{q,i}^{j|k}$ is \emph{tight}, otherwise, the crucial constraint is \emph{slack}, even though the constraint may be tight actually, that's $$x_{q,i}^{j|k}= \sum_{m=1}^Q\sum_{s=1}^{i-1}\frac{1}{s}\sum_{l=1}^K\left(x_{m,s}^{j-1|l}-x_{m,s}^{j|l}\right)+\sum_{m=1}^{q-1}\frac{1}{i}\sum_{l=1}^K\left(x_{m,i}^{j-1|l}-x_{m,i}^{j|l}\right)=0.$$

\noindent{\bf Dual Program}\ \ Suppose $b_i^k = \sum_{l=k}^Kw_l\frac{\binom{i-1}{k-1}\binom{n-i}{l-k}}{\binom{n-1}{l-1}}$. We have the dual program:
{\small
\begin{equation*}
\begin{split}
    &\qquad\min z =\sum_{q=1}^Q \sum_{k=1}^K\sum_{i=1}^{n/Q} y_{q,i}^{1|k}\\
    &\text{s.t.}\left\{
    \begin{aligned}
        &y_{q,i}^{j|k}+\frac{1}{i}\sum_{m=1}^Q\sum_{s=i+1}^{n/Q}\sum_{l=1}^K\left(y_{m,s}^{j|l}-y_{m,s}^{j+1|l}\right)+\frac{1}{i}\sum_{m=q+1}^Q\sum_{l=1}^K\left(y_{m,i}^{j|l}-y_{m,i}^{j+1|l}\right)\geq
		\frac{b_i^k}{nW},&\\
         &(1\leq q\leq Q,\,1\leq i\leq n/Q,\,1\leq j\leq J,\,1\leq k\leq K)\\
        &y_{q,i}^{j|k}\geq 0, ~~(1\leq q\leq Q,\,1\leq i\leq n/Q,\,1\leq j\leq J,\,1\leq k\leq K).\\
    \end{aligned}
    \right.
\end{split}
\end{equation*}
} 

In this program, we add a set of dummy variables $y_{q,i}^{(J+1)|k}$ and set them to be 0 for brief.
    Respectively, we can define the \emph{crucial constraint} and \emph{crucial position} for the $y_{q,i}^{j|k}$ and the \emph{$(0,1)$-solution} for this dual program.
\vspace{-0.1cm}
\subsection{Protocol Description}
\vspace{-0.1cm}
\SetKwInOut{Input}{input}\SetKwInOut{Output}{output}
\begin{algorithm}
    \caption{{\bf Preprocessing Part}\label{sec3:alg:sep}}
	\Input {$n$, $J$, $K$, $Q$, $\{w_k\,|\,1\leq k\leq K\}$}
	\Output {\{$\iqjk{q}{j}{k}\,|\,1\leq q\leq Q,\,1\leq j\leq J,\,1 \leq k\leq K$\}} 
	$\iqjk{q}{j}{k}\,(1\leq q\leq Q,\,1\leq j\leq J,\,1\leq k\leq K)$: $JKQ$ crucial positions, initially 1\\
	$\yqijk{q}{i}{j}{k}\ (1\leq q\leq Q,\ 1\leq i \leq n/Q,\ 1\leq j\leq J+1,\ 1\leq k\leq K)$: initially 0\\
	\For{$i= n/Q$ \text{to} $1$}{
		\For{ $q=Q$ \KwTo $1$}{
			\For{ $j=J$ \KwTo $1$}{
				\For{ $k = K$ \KwTo $1$}{
					\label{alg:gety1}
					{\small 
					$\yqijk{q}{i}{j}{k} \gets\frac{b_i^k}{nW}+\frac{1}{i}\sum\limits_{m=1}^Q\sum\limits_{s=i+1}^{n/Q}\sum\limits_{l=1}^K \left(\yqijk{m}{s}{j+1}{l} -\yqijk{m}{s}{j}{l}\right)+ \frac{1}{i} \sum\limits_{m = q+1}^{Q}\sum\limits_{l=1}^{K}\left(\yqijk{m}{i}{j+1}{l}-\yqijk{m}{i}{j}{l} \right)$\label{cal}}\\
					\If {$\yqijk{q}{i}{j}{k} \leq 0$ }{ 
						$\yqijk{q}{i}{j}{k} \gets 0$\\
						\If{$i=n$ \textbf{\em or} $\yqijk{q}{i+1}{j}{k} > 0$}
						{
							$\iqjk{q}{j}{k} \gets i+1$ \Comment{Find and record the crucial position}
						}
					}
				}
			}
		}
	}
\end{algorithm}

The protocol consists of two parts. The first part (Algorithm \ref{sec3:alg:sep}) takes $J$, $K$, $Q$ and $n$ as inputs and outputs $JKQ$ positions \{$\iqjk{q}{j}{k}\,|\,1\leq q\leq Q,\ 1\leq j\leq J,\ 1 \leq k\leq K$\}. We will show some properties about these positions later.
The preprocessing part actually solves the dual program as defined in
Section \ref{sec:name:lpmodel}. But it is more efficient than the ordinary LP solver. 
It is easy to check if we calculate the value of $y_{q,i}^{j|k}$ in line 7 carefully, the time complexity of the algorithm is $O(nJK^2)$.

The second part (Algorithm \ref{sec3:alg:aosp}) takes the output of
preprocessing part as input and does the interview on $Q$ queues
simultaneously. For each queue, this protocol consist of $J$ rounds. When $j$ ($1\leq j\leq J-1$) persons were selected from all queues, the protocol will enter the $(j+1)$-th round immediately.
In each round, the protocol divided candidates in each queue into $K+1$ phases. For each queue, in the $k$-th ($1\leq k\leq K$) phase, that's from $(\iqjk{q}{j}{k-1})$-th candidate to $(\iqjk{q}{j}{k}-1)$-th candidate, the protocol selects the $(k-1)$-th best person of previous $k-1$ phases in this queue as criteria, and just hires the first one that better than this criteria. Candidates in each queue come up one by one. For each candidate, the employers check the number of candidates selected to determine the current round, and then query the current phase based on the position of current candidate, and finally make decision by comparing with criteria of this phase.
The protocol will terminate when all candidates were interviewed or $J$
candidates are selected. In the protocol, we define a \emph{global} order
which is consistent with the problem definition. Using $c_{q,i}$ to stand for
the $i$-th candidate of $q$-th queue. We say $c_{q',i'}$ comes \emph{before} $c_{q,i}$ if $i' < i$ or $i' = i$ and $q'< q$. 

\begin{algorithm}
	\label{sec3:alg:aosp}
	\caption{{\bf Adaptive Observation-Selection Protocol}}
	\Input {$n$, $Q$, $J$, $K$, $\{\iqjk{q}{j}{k}\,|\, 1\leq q\leq Q,\, \,1\leq j\leq J,\, 1\leq k \leq K\}$}
	\Output{ the selected persons }
	let $i_{q,j,K+1}$ to be $n+1$ \\
	\For{ all queues simultaneously}{
		{$\triangleright$ suppose $q$ is the number of an arbitrary queue}\\
		\textbf{for} $i = 1$ \KwTo $\iqjk{q}{1}{1}-1$ do interview without selecting anyone\\
		\For {$i = \iqjk{q}{1}{1}$ \KwTo $n/Q$}{
			interview current candidate $c_{q,i}$\\
			let $j$ to be the number of selected persons \emph{before} $c_{q,i}$ in \emph{global order}\\
			\textbf{if} $j = J$ \textbf{then}
				\Return\\
				let $k$ to be the current phase number of $(j+1)$-th round \Comment{that's the $k$ satisfies $\iqjk{q}{j+1}{k-1} \leq i < \iqjk{q}{j+1}{k}$}\\
			let $s$ to be the $(k-1)$-th best one from the first candidate to $(\iqjk{q}{j+1}{k-1})$-th candidate\\
			\If { $c_{q,i}$ is better than $s$}{
				select $c_{q,i}$\\
			}
		}
	}
\end{algorithm}
%
\vspace{-0.2cm}
\subsection{Optimality of the Adaptive Observation-Selection Protocol}
\vspace{-0.2cm}
In the rest of this work, we use $y_{q,i}^{j|k*}$
to stand for the value of $\yqijk{q}{i}{j}{k}$ obtained from the preprocessing part for $1\leq q\leq Q,\,1\leq i\leq n/Q,\,1\leq j\leq J+1,\,1\leq k\leq K$. These two notations $y_{q,i}^{j|k}$ and $y_{q,i}^{j|k*}$ should be clearly distinguished. The former is a variable in the dual program, while the latter is a value we get from the preprocessing part. 
\vspace{-0.3cm}
\subsubsection{Preparations} 
For the clarity of the proof, we distill some fundamental results in this part. 
The Proposition \ref{sec3:prop:bik} talks about two properties of $b_i^k$ defined in the dual program, and the Lemma \ref{sec3:prop:iqjk}, \ref{sec3:prop:jgeq} reveal some important properties of the preprocessing part. The Lemma \ref{sec3:prop:recurrence} considers a recurrence pattern. This recurrence can be used to explore the structure of the constraints of the dual program.

\def\propbik{
	For $1\leq k\leq K,\ 1\leq i\leq n/Q$, $b_i^k$ satisfies \textbf{\emph{(a)}} $ib_i^k \leq (i+1)b_{i+1}^k$ and \textbf{\emph{(b)}} $b_i^k \geq b_i^{k+1}$.
}
\vspace{-0.1cm}
\begin{prop}
	\label{sec3:prop:bik}
	\propbik
\end{prop}

\begin{proof}
    \emph{\textbf{a.}} According to the definition of $b_i^k$, we have
	{\small
    \begin{equation*}
    \begin{split}
        &ib_i^l - (i+1)b_{i+1}^l\\
        =&\sum_{k=l}^K w_k \left(i\frac{\binom{n-i}{k-l}\binom{i-1}{l-1}}
            {\binom{n-1}{k-1}}-(i+1)\frac{\binom{n-i-1}{k-l}\binom{i}{l-1}}
            {\binom{n-1}{k-1}}\right)\\
        =&\sum_{k=l}^K \frac{w_k(k-1)!(n-k)!(n-i-1)!\,i!}
                        {(n-1)!(k-l)!(n-i-1-k+l)!(l-1)!(i-l)!}
             \left(\frac{n-i}{n-i-k+l}-\frac{i+1}{i-l+1}\right)\\
        =&\frac{1}{\binom{n-1}{i}}\sum_{k=l}^K \frac{w_k(k-1)!(n-k)!}
                        {(k-l)!(n-i-k+l)!(l-1)!(i-l+1)!}\left((i+1)k-(n+1)l\right)\\
        =&\frac{1}{\binom{n-1}{i}(l-1)!(i-l+1)!}\sum_{k=l}^K \frac{w_k(k-1)!(n-k)!}
                        {(k-l)!(n-i-k+l)!}\left((i+1)k-(n+1)l\right).
    \end{split}
    \end{equation*}
}
 Let $s_k = \frac{(k-1)!(n-k)!}{(k-l)!(n-i-k+l)!}\left((i+1)k-(n+1)l\right)$. We only need to prove $\sum_{k =l}^K w_ks_k$ is non-positive as the rest part of above expression is always positive.

The sign of $s_k$ is determined by the part $((i+1)k-(n+1)l)$ which is increasing when $k$ increases.
When $k = l$, $s_k \leq 0$ due to $ i \leq n$. Let $k'$ stand for the maximum $k$ that makes $s_k\leq 0$. That's to say, we have
    \begin{equation*}
        s_k \ \left\{
            \begin{aligned}
            &\leq 0,& &l\leq k\leq k'&\\
            &> 0,& &k'<k\leq K.&
            \end{aligned}\right.
    \end{equation*}
Because $w_1 \geq w_2\geq \cdots \geq w_{k'}\geq \cdots \geq w_K$, so we have
    $\sum_{k = l}^{k'} w_ks_k\leq \sum_{k=l}^{k'}w_{k'}s_k$,
and
    $\sum_{k = k'+1}^K w_ks_k\leq \sum_{k=k'+1}^K w_{k'}s_k$.
Thus
    \begin{equation*}
        \sum_{k=l}^K w_ks_k \leq \sum_{k=l}^K w_{k'}s_k =
        w_{k'}\sum_{k=l}^K s_k.
    \end{equation*}
Let $S_K = \sum_{k=l}^K s_k$, then it is sufficient to prove $S_K\leq 0$.

    Next, we prove
    \begin{equation}
    \label{equ}
    \begin{split}
        S_K
        =\frac{-K!(n-K)!}{(K-l)!(n-K-i+l-1)!}.
    \end{split}
    \end{equation}
\par
    Fix $l$, and we use induction on $K$ to prove it. The basis case is $K=l$. We have $S_K = s_l = \frac{-l!(n-l)!}{(n-i-1)!}$, which satisfies the Equation \ref{equ}. Suppose Equation~\ref{equ} is held for ${K-1}$. We have
    \begin{equation*}
    \begin{split}
          S_K 
        = &S_{K-1}+\frac{(K-1)!(n-K)!}
                        {(K-l)!(n-i-K+l)!}\left((i+1)K-(n+1)l\right)\\
        = &\frac{-(K-1)!(n-K+1)!}{(K-1-l)!(n-K-i+l)!}+
            \frac{(K-1)!(n-K)!}{(K-l)!(n-i-K+l)!}\left((i+1)K-(n+1)l\right)\\
        = &\frac{(K-1)!(n-K)!}{(K-l-1)!(n-i-K+l)}\left(\frac{(i+1)K-(n+1)l}
            {K-l}-n+K-1\right)\\
        = &\frac{(K-1)!(n-K)!}{(K-l-1)!(n-i-K+l)}\cdot
            \frac{-K(n-i-K+l)}{K-l}\\
        = &\frac{-K!(n-K)!}{(K-l)!(n-K-i+l-1)!}.
    \end{split}
    \end{equation*}
	So the Equation \ref{equ} is true and we have $S_K\leq 0$. Consequently, it is true that $ib_i^l\leq (i+1)b_{i+1}^i$.

    \vspace{0.3cm}     
	\noindent\emph{\textbf{ b.}} Let the left part subtract the right part and we get {\small 
    \begin{equation*}
    \begin{split}
        &b_i^l - b_i^{l+1}\\
        =&\sum_{k=l}^K w_k \frac{\binom{n-i}{k-l}\binom{i-1}{l-1}}
                           {\binom{n-1}{k-1}} -
          \sum_{k=l+1}^K w_k \frac{\binom{n-i}{k-l-1}\binom{i-1}{l}}
                           {\binom{n-1}{k-1}}\\
        =&\frac{w_l(i-1)!(n-l)!}{(i-l)!(n-1)!}+
            \frac{(k-1)!(n-k)!}{(n-1)!}\sum_{k=l+1}^K w_k \left(
                \binom{n-i}{k-l}\binom{i-1}{l-1}
                -\binom{n-i}{k-l-1}\binom{i-1}{l}\right)\\
        =&\frac{(i-1)!}{(i-l)!(n-1)!}\left(w_l(n-l)!+\frac{(n-i)!}{l!}
            \sum_{k=l+1}^K\frac{w_k(k-1)!(n-k)!(ln+l-ik)}{(k-l)!(n-i-k+l+1)!}
            \right).
    \end{split}
    \end{equation*}
}

    Firstly, we show the following equation
	{\small
    \begin{equation}
    \label{equ1}
            \left((n-l)!+\frac{(n-i)!}{l!}
                \sum_{k=l+1}^K\frac{(k-1)!(n-k)!(ln+l-ik)}{(k-l)!(n-i-k+l+1)!}
                \right) =\frac{(n-i)!(n-K)!K!}{(n-i-K+l)!\,l!\,(K-l)!}.
    \end{equation}
}

We use induction on $K$ to prove it. The basis is the case when $K=l+1$: both the left part and the right part of Equation \ref{equ1} are
$(n-l-1)!(n-i)(l+1)$.  So the Equation \ref{equ1}  is held for $K=l+1$. Then, for general $K>l+1$, we assume that the Equation~\ref{equ1} is held for $K-1$. We have
    \begin{equation*}
    \begin{split}
        &(n-l)!+\frac{(n-i)!}{l!}\sum_{k=l+1}^K
        \frac{(k-1)!(n-k)!(ln+l-ik)}{(k-l)!(n-i-k+l+1)!}\\
        =&\frac{(n-i)!(n-K+1)!(K-1)!}{(n-i-K+1+l)!\,l!\,(K-l-1)!}
            + \frac{(n-i)!(K-1)!(n-K)!(ln+l-iK)}{l!\,(K-l)!(n-i-K+l+1)!}\\
        =&\frac{(n-i)!(K-1)!(n-K)!}{(n-i-K+l+1)!\,l!\,(K-l-1)!}
            \left(n-K+1+\frac{ln+l-iK}{K-l}\right)\\
        =&\frac{(n-i)!(n-K)!K!}{(n-i-K+l)!\,l!\,(K-l)!}.
    \end{split}
    \end{equation*}
Thus, by induction, the Equation~\ref{equ1} is held.

Let $s_k$ stand for $\frac{(k-1)!(n-k)!(ln+l-ik)}{(k-l)!(n-i-k+l+1)!}$, then the sign of $s_k$ depends on the sign of $(ln+l-ik)$ which is decreasing as $k$ increases. Let $k'\ (k' \geq l+1)$ stand for the maximum $k$ such that $s_k$ is non-negative. It means
that
    \begin{equation*}
        s_k \ \left\{
            \begin{aligned}
            & > 0,& &l < k\leq k'&\\
            & \leq 0,& &k'<k\leq K.&
            \end{aligned}\right.
    \end{equation*}
As $w_k$ is non-increasing with $k$ goes up, we have
\begin{equation*}
    \begin{split}
        &b_i^l - b_i^{l+1}\\
        =&\frac{(i-1)!}{(i-l)!(n-1)!}\left(w_l(n-l)!+\frac{(n-i)!}{l!}
            \sum_{k=l+1}^K\frac{w_k(k-1)!(n-k)!(ln+l-ik)}{(k-l)!(n-i-k+l+1)!}
            \right)\\
        =&\frac{(i-1)!}{(i-l)!(n-1)!}\left(w_l(n-l)!+\frac{(n-i)!}{l!}
            \sum_{k=l+1}^K w_k s_k\right)\\
        \geq & \frac{(i-1)!}{(i-l)!(n-1)!}\left(w_{k'}(n-l)!+\frac{(n-i)!}{l!}
            \sum_{k=l+1}^K w_{k'}s_k\right)\\
        =& \frac{(i-1)!}{(i-l)!(n-1)!}\cdot w_{k'}\cdot \frac{(n-i)!(n-K)!K!}{(n-i-K+l)!\,l!\,(K-l)!}\\
        \geq&0.
    \end{split}
\end{equation*}
Thus, we finish the proof.
   \qed
\end{proof}

\def\propiqjk{
The $\{\iqjk{q}{j}{k}\,|\,1\leq q\leq Q,\, 1\leq j\leq J,\, 1\leq k\leq K\}$ obtained from the preprocessing part satisfies $\iqjk{q}{j}{t}\leq \iqjk{q}{j}{k}$,
and  we have $y_{q,i}^{j|t*} \geq y_{q,i}^{j|k*} > 0$ for $1 \leq t < k$.}
\vspace{-0.3cm}
\begin{lemm}
	\label{sec3:prop:iqjk}
	\propiqjk
\end{lemm}

\begin{proof}
	Note that the several proofs including this one heavily depend on a \emph{key observation} that $$y_{q,i}^{j|k*}\geq \frac{b_i^{k}}{nW} - \frac{1}{i}\sum_{m=1}^Q\sum_{s=i+1}^n\sum_{l=1}^K\left(y_{m,s}^{j|l*}-y_{m,s}^{(j+1)|l*}\right)-\frac{1}{i}\sum_{m=q+1}^Q\sum_{l=1}^K\left(y_{m,i}^{j|l*}-y_{m,i}^{(j+1)|l*}\right)
	$$ is always true according to the preprocessing part, and the left side and right side must be equal if $y_{q,i}^{j|k*} > 0$.

	In the preprocessing part, $\iqjk{q}{j}{k}$ records the crucial position that the value of $y_{q,i}^{j|k*}$ transforms from positive to zero. That is to say $y_{q,\iqjk{q}{j}{k}}^{j|k*} > 0$ while $y_{q,\iqjk{q}{j}{k}-1}^{j|k*}=0$.

When $y_{q,i}^{j|k*} > 0$, according to the key observation mentioned above, we have
    \begin{equation*}
    \begin{split}
		&y_{q,i}^{j|t*}-y_{q,i}^{j|k*}\\
		\geq & \frac{b_i^{t}}{nW} - \frac{1}{i}\sum_{m=1}^Q\sum_{s=i+1}^n\sum_{l=1}^K\left(y_{m,s}^{j|l*}-y_{m,s}^{(j+1)|l*}\right)-\frac{1}{i}\sum_{m=q+1}^Q\sum_{l=1}^K\left(y_{m,i}^{j|l*}-y_{m,i}^{(j+1)|l*}\right)\\
        &-\frac{b_i^{k}}{nW} + \frac{1}{i}\sum_{m=1}^Q\sum_{s=i+1}^n\sum_{l=1}^K\left(y_{m,s}^{j|l*}-y_{m,s}^{(j+1)|l*}\right)+\frac{1}{i}\sum_{m=q+1}^Q\sum_{l=1}^K\left(y_{m,i}^{j|l*}-y_{m,i}^{(j+1)|l*}\right)\\
		=& \frac{b_i^{t}}{nW}- \frac{b_i^{k}}{nW}
        \geq 0.
    \end{split}
    \end{equation*}
	The last inequality is due to Proposition~\ref{sec3:prop:bik}.b. Thus, we have $y_{q,i}^{j|t*}\geq y_{q,i}^{j|k*}$.

	When $y_{q,i}^{j|k*}=0$, it is obvious that $y_{q,i}^{j|(k-1)*}\geq y_{q,i}^{j|k*}$ as the preprocessing part always assigns a non-negative value to $y_{q,i}^{j|t*}$. So, $\yqijk{q}{i}{j}{t*}$ is always no less than $\yqijk{q}{i}{j}{k*}$. This implies $\iqjk{q}{j}{t} \leq \iqjk{q}{j}{k}$.
\qed
\end{proof}

%
\def\propjgeq{
	According to the preprocessing part, if $y_{q,i}^{j|k*} > 0$ and $y_{q,i}^{j|k*}\geq y_{q,i}^{j+1|k*}$, we have $y_{q,i}^{j|t*}\geq y_{q,i}^{j+1|t*}$ for $1\leq t < k$.
}
\vspace{-0.2cm}
\begin{lemm}
	\label{sec3:prop:jgeq}
	\propjgeq
\end{lemm}

\begin{proof}
	According to Lemma \ref{sec3:prop:iqjk}, we have $y_{q,i}^{j|t*}\geq y_{q,i}^{j|k*}>0$. It is clear that 
	\begin{equation}
		\label{sec3:temp:eq1}
		y_{q,i}^{j|t*} - y_{q,i}^{j|k*} = \frac{b_i^t-b_i^k}{nW}.
	\end{equation}
	When $y_{q,i}^{(j+1)|t*}> 0$, as $y_{q,i}^{(j+1)|k*}$ may be equal to 0, we have
{\small
	\begin{equation}
		y_{q,i}^{(j+1)|t*}+\frac{1}{i}\sum_{m=1}^Q\sum_{s=i+1}^{n/Q}\sum_{l=1}^K\left(y_{m,s}^{(j+1)|l*}-y_{m,s}^{(j+2)|l*}\right)+\frac{1}{i}\sum_{m=q+1}^Q\sum_{l=1}^K\left(y_{m,i}^{(j+1)|l*}-y_{m,i}^{(j+2)|l*}\right) = \frac{b_i^t}{nW},
	\end{equation}
	\begin{equation}
		y_{q,i}^{(j+1)|k*}+\frac{1}{i}\sum_{m=1}^Q\sum_{s=i+1}^{n/Q}\sum_{l=1}^K\left(y_{m,s}^{(j+1)|l*}-y_{m,s}^{(j+2)|l*}\right)+\frac{1}{i}\sum_{m=q+1}^Q\sum_{l=1}^K\left(y_{m,i}^{(j+1)|l*}-y_{m,i}^{(j+2)|l*}\right) \geq \frac{b_i^k}{nW}.
	\end{equation}
}
	Then we can get 
	\begin{equation}
		\label{sec3:temp:eq2}
		y_{q,i}^{(j+1)|t*} - y_{q,i}^{(j+1)|k*} \leq \frac{b_i^t-b_i^k}{nW}.
	\end{equation}
	From Equation \ref{sec3:temp:eq1} and \ref{sec3:temp:eq2} and the condition that $y_{q,i}^{j,k*} \geq y_{q,i}^{(j+1)|k*}$, we have $y_{q,i}^{j|l*}\geq y_{q,i}^{(j+1)|l*}$. This lemma is certainly true when $y_{q,i}^{(j+1)|t*} = 0$.
	\qed
\end{proof}


\def\proprecurrence{
	Suppose $m,\ t,\ Q,\ K$ are positive integers and $c$ is a constant real number. $\{f_t\}_{t=1}^m$, $\{g_t\}_{t=1}^m$ and $\{h_t\}_{t=1}^m$ are three sequences. Let $i=\floor{\frac{t-1}{Q}}+1$, if the recursion 
	$f_t+\frac{K}{i}\sum_{s=t+1}^m(f_s-g_s)+\frac{c}{i}=h_{i}$ is held, then all the values in $\{f_t\}_{t=1}^n$ will increase when $c$ decreases or values in $\{g_t\}_{t=1}^n$ increase.
}
\vspace{-0.2cm}
\begin{lemm}
	\label{sec3:prop:recurrence}
	\proprecurrence
\end{lemm}

\begin{proof}
	Let $I = \floor{\frac{m-1}{Q}}+1$, thus $1\leq i\leq I$. The $\{f_t\}_{t=1}^n$ can be divided into $I$ segments. The $i$-th segment is when $(i-1)Q+1\leq t\leq iQ$. We use induction on $i$ to show the lemma is held on every segment.
    
	Firstly, we consider the $I$-th segment, that's when $(I-1)Q+1\leq t\leq m$. In this segment, the value of $\floor{\frac{t-1}{Q}}+1$ is fixed to be $I$. Thus the recursion can be rewritten as
    \begin{equation}
        f_t+\frac{K}{I}\sum_{s=t+1}^m(f_s-g_s)+\frac{c}{I}=h_{I}.\label{sec3:equ:tempforI}
    \end{equation}
    From the view of $t$, $I$ can be considered as a constant integer. There has
    $$f_t-f_{t+1}=\frac{K}{I}(g_{t+1}-f_{t+1})\Rightarrow f_t = \frac{I-K}{I}f_{t+1}+\frac{K}{I}g_{t+1}.$$
    Finally, we get 
    \begin{equation*}
    \begin{split}
        f_t &= \left(\frac{I-K}{I}\right)^{m-t}f_m +\frac{K}{I}\sum_{s = t+1}^m\left(\frac{I-K}{I}\right)^{s-t-1}g_s\\
            &= \left(\frac{I-K}{I}\right)^{m-t}\left(h_I-\frac{c}{I}\right) +\frac{K}{I}\sum_{s = t+1}^m\left(\frac{I-K}{I}\right)^{s-t-1}g_s.
    \end{split}
    \end{equation*}
	Note that the coefficient of $c$ is negative and that of $g_s$ is positive, So in this interval, that's in the $I$-th segment, the lemma is true. We define a sequence $\{A_i\}_{i=1}^I$, and $A_i$ satisfies 
    \begin{equation}
		A_i+\frac{K}{I}\sum_{s=(i-1)Q+1}^m(f_s-g_s)+\frac{c}{I}=h_{i}\label{sec7:temp1:2},
    \end{equation}
    Through the same procedure for the analysis of $f_t$, we can conclude that $A_I$ have the same property decribed in this lemma. We say $A_I$ is compatible.
    
	Using induction, we suppose from $(i+1)$-th segment to $I$-th segment, that's $iQ+1\leq t\leq m$, this lemma is true and all $A_l$, $i+1\leq l\leq I$, are compatible. Our target is to show for $i$-th segment, that's when $(i-1)Q+1\leq t\leq iQ$, this lemma still held and $A_i$ is compatible too.
    
    Let $P_t = K\sum_{s=t}^m(f_s-g_s)+c$. From the equation
    \begin{equation}
        A_{i+1}+\frac{K}{i+1}\sum_{s=iQ+1}^m(f_s-g_s)+\frac{c}{i+1}=h_{i+1},\label{sec:temp1:2}
    \end{equation} 
    we can obtain that
    $$P_{iQ+1} = (i+1)h_{i+1}-(i+1)A_{i+1}.$$
	The value of $P_{iQ+1}$ will decrease when  decreasing $c$ or increasing values in $\{g_t\}_{t=1}^n$ because these operations will make $A_{i+1}$ increase. 
    
    Then, when $(i-1)Q+1\leq t\leq iQ$, we have 
    \begin{eqnarray}
		&&f_t+\frac{K}{i}\sum_{s=t+1}^m(f_s-g_s)+\frac{c}{i}=h_{i}\label{sec7:temp1:3}\\
			  &\iff & f_t + \frac{K}{i}\sum_{s=t+1}^{iQ}(f_s-g_s)+\frac{P_{iQ+1}}{i} = h_i\label{sec3:equ:genei}.
    \end{eqnarray}
	Comparing Equation \ref{sec3:equ:tempforI} and Equation \ref{sec3:equ:genei}, we find they are of the same form as $P_{iQ+1}$ can be considered as a constant. Using the same method as when $i=I$, we get the conclusion that $\{f_t\}_{t=1}^n$ will increase when values in $\{g_t\}_{t=1}^n$ increase or $P_{iQ+1}$ decreases.  Taking the relationship between $P_{iQ+1}$, $\{g_t\}_{t=1}^n$ and $c$ into consideration, we know the lemma is held for $i$-th segment, that's $(i-1)Q+1\leq t \leq iQ$. Besides, it is easy to show $A_i$ still keeps compatible using the same method as when $i= I$. Using induction, we finish the proof.\qed
\end{proof}

\subsubsection{Main Frame of the Proof}
The main idea of the proof is described as follow. Firstly we show the fact that the Adaptive Observation-Selection protocol can be mapped to a feasible $(0,1)$-solution of the primal program (Lemma \ref{sec3:lemm:algtoprimal}) while the $\{\yqijk{q}{i}{j}{k*}\,|\,1\leq q\leq Q,\ 1\leq i\leq n/Q,\ 1\leq j\leq J,\ 1\leq k\leq K\}$ obtained from the preprocessing part is corresponding to a feasible $(0,1)$-solution of the dual program (Lemma \ref{sec3:lemm:01ofdual}). Then, we argue that these two feasible $(0,1)$-solutions satisfy the \emph{theorem of complementary slackness} (Theorem \ref{sec3:theo:optimality}). Thus both the solutions are optimal respectively. This means our protocol is optimal. 

%
\def\lemmalgtoprimal{
	Taking the \{$\iqjk{q}{j}{k}\,|\,1\leq q\leq Q,\,1\leq j\leq J,\,1 \leq k\leq K$\} obtained from the preprocessing part as input, the Adaptive Observation-Selection Protocol can be mapped to a $(0,1)$-solution of the primal program and the $i_{q,j,k}$ is the crucial position of $x_{q,i}^{j|k}$.
}
\vspace{-0.2cm}
\begin{lemm}
	\label{sec3:lemm:algtoprimal}
	\lemmalgtoprimal
\end{lemm}

\begin{proof}
	As mentioned before, we use $c_{q,i}$ to stand for the $i$-th candidate in $q$-th queue, and we say $c_{q',i')}$ is before $c_{q,i}$ if $i'<i$ or $i'=i$ and $q'<q$. Besides, we use $A_{q,i}^{j|k}$ to stand for the event that $c_{q,i}$ is selected in the $j$-th round (that's the $j$-th one selected in all queue) given that he/she is the $k$-th best from $1$ to $i$ in queue $q$ by the protocol. That's, $x_{q,i}^{j|k}=\Pr(A_{q,i}^{j|k})$. According to the Adaptive Observation-Selection protocol, when $i < \iqjk{q}{j}{k}$, we know $\Pr(A_{q,i}^{j|k})=0$, however, when $i \geq \iqjk{q}{j}{k}$, $c_{q,i}$ must be selected given he/she is the $k$-th best up to now. Thus the event $A_{q,i}^{j|k}$ happens is reduced to exact $j-1$ persons were hired in all queues before $c_{q,i}$. Denote $B_{q,i}^j$ as the event that there are at least $j$ persons selected before person $c_{q,i}$ in all queue and $C_{q,i}^j$ as the event that $c_{q,i}$ was selected in $j$-th round. Thus we have
    \begin{equation*}
        \begin{split}
            \Pr(A_{q,i}^{j|k}) &= \Pr(\text{there are just $j-1$ persons that selected before $c_{q,i}$})\\
            &= \Pr(B_{q,i}^{j-1})-\Pr(B_{q,i}^j)\\
            &= \sum_{m=1}^Q\sum_{s=1}^{i-1}\left(\Pr(C_{m,s}^{j-1})-\Pr(C_{m,s}^j)\right)+\sum_{m=1}^{q-1}\left(\Pr(C_{m,i}^{j-1}-\Pr(C_{m,i}^j)\right).\\
        \end{split}
    \end{equation*}
	On the other hand, we have 
	\begin{equation*}
		\begin{split}
		\Pr(C_{q,i}^j) &= \sum_{k=1}^K \Pr(c_{q,i}\text{ is select and he is $k$-th best up to now })\\
		&=\frac{1}{i}\sum_{l=1}^Kx_{q,i}^{j|l}.
		\end{split}
	\end{equation*}
	Combining above results, we get
{\small
    \begin{equation*}
        x_{q,i}^{j|k} = \left\{
        \begin{aligned}
            &0,&&1 \leq i <\iqjk{q}{j}{k}&\\
            &\sum_{m=1}^Q\sum_{s=1}^{i-1}\frac{1}{s}\sum_{l=1}^K\left(x_{m,s}^{j-1|l}-x_{m,s}^{j|l}\right)+\sum_{m=1}^{q-1}\frac{1}{i}\sum_{l=1}^K\left(x_{m,i}^{j-1|l}-x_{m,i}^{j|l}\right),& &\iqjk{q}{j}{k}\leq i\leq n/Q.&\\
        \end{aligned}
        \right.
    \end{equation*}
}
    This is the definition of the $(0,1)$-problem, and we can see $\iqjk{q}{j}{k}$ is the crucial position for $x_{q,i}^{j|k}$.
    \qed
\end{proof}

The multiple queues contribute lots of complexity to the dual program. Before the proof of Lemma \ref{sec3:lemm:01ofdual} , we provide a closely relative proposition to simplify the LP model.
\begin{prop}
\label{sec3:prop:qqueuesto1}
    The crucial constraint in the dual program 
    \begin{equation}
        \begin{aligned}
        &y_{q,i}^{j|k}+\frac{1}{i}\sum_{m=1}^Q\sum_{s=i+1}^{n/Q}\sum_{l=1}^K\left(y_{m,s}^{j|l}-y_{m,s}^{(j+1)|l}\right)+\frac{1}{i}\sum_{m=q+1}^Q\sum_{l=1}^K\left(y_{m,i}^{j|l}-y_{m,i}^{(j+1)|l}\right)\geq \frac{b_i^k}{nW},&\\
         &(1\leq q\leq Q,\,1\leq i\leq n/Q,\,1\leq j\leq J,\,1\leq k\leq K)
         \end{aligned}
    \end{equation}
         is equivalent to the inequality
    \begin{equation}
        \begin{aligned}
			&y_t^{j|k}+\frac{1}{\floor{\frac{t-1}{Q}}+1}\sum_{s=t+1}^{n}\sum_{l=1}^K\left(y_s^{j|l}-y_s^{(j+1)|l}\right)\geq \frac{b_{\floor{\frac{t-1}{Q}}+1}^k}{nW},&\\
         &(1\leq t\leq n,\,1\leq j\leq J,\,1\leq k\leq K)
        \end{aligned}
    \end{equation}
    with the relationship $y_t^{j|k} = y_{q,i}^{j|k}$ where $t=Qi+q$.
\end{prop}

    This proposition is obtained by merging the $Q$ queues into a single queue according to the order candidates come, that's a sequence as $$c_{1,1},c_{2,1},\cdots,c_{Q,1},c_{1,2},c_{2,2},\cdots,c_{Q,2},\cdots,c_{1,n/Q},c_{2,n/Q},\cdots,c_{Q,n/Q}.$$ As we can see, the relationship between $y_t^{j|k}$ and $y_{q,i}^{j|k}$ is a bijection. All properties mentioned before for $y_{q,i}^{j|k}$ are still held for $y_t^{j|k}$.

The relationship between the preprocessing part and the dual program is the essential and most complicate part in this work. As the dual program is extremely complex, insight on the structure should be raised. The proof relies heavily on the properties of the preprocessing part and the dual program revealed in preparation part.
\def\lemmalgtodual{
The $\{y_{q,i}^{j|k*}\,|\,1\leq q\leq Q,\, 1\leq i\leq n/Q,\,1\leq j\leq J,\, 1\leq k\leq K\}$ obtained from the preprocessing part is a $(0,1)$-solution of the dual program.
}
\begin{lemm}
    \label{sec3:lemm:01ofdual}
	\lemmalgtodual
\end{lemm}

\begin{proof}
    At first, we prove that the $y_{q,i}^{j|k*}$ is a feasible solution. From the preprocessing part, it is easy to show that the $y_{q,i}^{j|k*}$ satisfies the non-negative constraint. When $y_{q,i}^{j|k*} > 0$, we know the crucial constraint for $y_{q,i}^{j|k}$ is tight. When keeping tight makes $y_{q,i}^{j|k*} < 0$, setting it to zero will still satisfy the crucial constraints and make the crucial constraint of $y_{q,i}^{j|k}$ slack. So, the $y_{q,i}^{j|k*}$ satisfies the crucial constraint. Thus we just need to show that $y_{q,i}^{j|k*}$ has $(0,1)$ property.
    
	Considering the Proposition \ref{sec3:prop:qqueuesto1}. Let $y_t^{j|k*} = y_{q,i}^{j|k*}$ where $t = Qi+q$, and use $i_t$ stand for $\floor{\frac{t-1}{Q}}+1$ for concision. If we can show that there is a $t_{j,k}$ make $y_t^{j|k*}$ satisfy
    \begin{equation}
    \label{sec3:equ:yt01}
        y_t^{j|k*} = \left\{
        \begin{aligned}
			&\frac{b_{i_t}^k}{nW}-\frac{1}{i_t}\sum_{s=t+1}^{n}\sum_{l=1}^K\left(y_s^{j|l*}-y_s^{(j+1)|l*}\right),& &t_{j,k}\leq t\leq n&\\
            &0,& &1\leq t < t_{j,k}.&
        \end{aligned}
        \right.
    \end{equation}
    for all $j$ and $k$, it is sufficient to get the conclusion that $y_{q,i}^{j|k*}$ satisfies the $(0,1)$ property.
	We complete this proof by induction on $j$. The hypotheses of the induction are, for any $j$, $1\leq j\leq J$,
    \begin{enumerate}
        \item $y_t^{j|k*} \geq y_t^{(j+1)|k*}$ for $1\leq t\leq n,\,1\leq k\leq K$;
		\item There is a $t_{j,k}$ that makes $y_t^{j|k*}$ satisfy Equation \ref{sec3:equ:yt01} for all $1\leq k\leq K$.
    \end{enumerate}

	The basis is when $j=J$. The hypothesis 1, $y_t^{J|k*}\geq y_t^{J+1|k*}$, is held for any $t,\, k$ as $y_t^{J+1|k*}$ is set to 0. The hypothesis 2 can be shown based on hypothesis 1. According to preprocessing part, when $y_t^{J|k*} > 0$, we have the equation
	$$y_t^{J|k*} = \frac{b_{i_t}^k}{nW}-\frac{1}{i_t}\sum_{s=t+1}^{n}\sum_{l=1}^K\left(y_s^{J|l*}-y_s^{J+1|l*}\right).$$
     Multiplying $i_t$ on both sides, we get
	$$i_ty_t^{J|k*} = \frac{i_tb_{i_t}^k}{nW}-\sum_{s=t+1}^{n}\sum_{l=1}^K\left(y_s^{J|l*}-y_s^{J+1|l*}\right).$$
Considering the right part of above equation. When $t$ is going down, the first term, $\frac{ib_i^k}{nW}$, is non-increasing  due to the Proposition \ref{sec3:prop:bik}.a, while the second term 
is non-decreasing because $y_t^{J|k}>y_t^{J+1|k}$. Thus, the right part totally is monotone and non-increasing when $t$ goes down. Let $t$ keep going down, once $y_t^{J|k}$ is set to zero, for all $1\leq t'< t$, $y_{t'}^{J|k}$ will be set to zero by preprocessing part, because the left part must be non-positive.  That's to say, there must be a $t_{J,k}$ that makes $y_{t}^{J|k*}$ satisfy the Equation \ref{sec3:equ:yt01}. 
The $t_{J,k}$ is at least 1. 
Note that the procedure to show hypothesis 2 is independent on the value $J$ or $k$, that's to say, this proof works for any $j$. So we have the following fact.
    \begin{fact}
    \label{sec3:fact:temp1}
         For any $j$, $1\leq j\leq J$, if the hypothesis 1 is held for all $k$, $1\leq k\leq K$, the hypothesis 2 is held too.
    \end{fact}
    
	Now we begin the induction part, and assume the hypothesis 1 and 2 are held from $j+1$ to $J$ for any $t$ and $k$. The target is to show the hypotheses are also held for $j$. 
	Note that due to Fact \ref{sec3:fact:temp1}, we just need to show hypothesis 1 is held.
	
	To show the hypothesis 1 is held for $j$, we use induction on $k$ and the basis is the case $k=K$, that's to show $y_t^{j|K*}\geq y_t^{(j+1)|K*}$ for all $t$. 
	As we can see, for large enough $t$ (at most $t=n$), both $y_t^{j|K*}$ and $y_t^{(j+1)|K*}$ are greater than 0. By Lemma \ref{sec3:prop:iqjk}, $y_t^{(j+1)|k}$ and $y_t^{j|k}$ are greater than zero for $1\leq k\leq K$. Then we have
    \begin{eqnarray}
		&&y_t^{j|k*} + \frac{1}{i_t}\sum_{s=t+1}^n\sum_{l=1}^K(y_s^{j|l*}-y_s^{(j+1)|l*})=\frac{b_{i_t}^k}{nW},\label{sec3:equ:temp16}\\
	  &&y_t^{(j+1)|k*} + \frac{1}{i_t}\sum_{s=t+1}^n\sum_{l=1}^K(y_s^{(j+1)|l*}-y_s^{(j+2)|l*})=\frac{b_{i_t}^k}{nW}\label{sec3:equ:temp17}
    \end{eqnarray} for $1\leq k\leq K$.
	Let $r_t^{j|k} = \sum_{l=1}^k y_{t}^{j|l*}$ and $\beta_i^k = \sum_{l = 1}^k \frac{b_i^l}{nW}$. Add up the both sides of Equation \ref{sec3:equ:temp16} and Equation \ref{sec3:equ:temp17} for all $1\leq k\leq K$ respectively, we get
    \begin{eqnarray}
        &&r_{t}^{j|K} + \frac{K}{i_t}\sum_{s=t+1}^n(r_s^{j|K}-r_s^{(j+1)|K})=\beta_{i_t}^K\label{sec3:temp3:1},\\
        &&r_{t}^{(j+1)|K} + \frac{K}{i_t}\sum_{s=t+1}^n(r_s^{(j+1)|K}-r_s^{(j+2)|K})=\beta_{i_t}^K\label{sec3:temp3:2}.
    \end{eqnarray}
	It is not hard to see the above Equations \ref{sec3:temp3:1} and \ref{sec3:temp3:2} satisfy the recursion described in Lemma \ref{sec3:prop:recurrence}. Thus we have $r_{t}^{j|K*}\geq r_{t}^{(j+1)|K*}$ because of $r_s^{(j+1)|K}\geq r_s^{(j+2)|K}$  according to the induction hypothesis $y_t^{(j+1)|k*}\geq y_{q,i}^{(j+2)|k*}$ for $1\leq k \leq K$. On the other hand, through manipulation on Equations \ref{sec3:equ:temp16} to \ref{sec3:temp3:2}, we have
    \begin{eqnarray}
        &&y_t^{j|K*} = \frac{r_t^{j|K}-\beta_{i_t}^K}{K} + \frac{b_{i_t}^K}{nW},\\
        &&y_t^{(j+1)|K*} = \frac{r_{t}^{(j+1)|K}-\beta_{i_t}^K}{K} + \frac{b_{i_t}^K}{nW}.
    \end{eqnarray}
	Then, we know $y_{t}^{j|K*}\geq y_{t}^{(j+1)|K*}$ is held when $y_t^{(j+1)|K*}>0$. Recall that $y_t^{(j+1)|k}$, $1\leq k\leq K$, has the $(0,1)$ property due to the hypothesis. When $y_t^{(j+1)|K*} = 0$, that's for $1\leq t < t_{j+1,K}$, $y_t^{j|K*}\geq y_t^{(j+1)|K*}$ is held too, because $y_t^{j|K*}$ is always set to be non-negative by the preprocessing part. Thus we finish the proof for the basis $k=K$. 
	
	Then, we show for a general $k$, $y_t^{j|k*}\geq y_t^{(j+1)|k*}$ is held given that $y_t^{j|l*}\geq y_t^{(j+1)|l*}$ for $k+1\leq l\leq K$ by induction. 

	Denote the largest $t$ that makes $y_{t-1}^{j|(k+1)*}$ equal to 0 as $t_{j,k+1}'$. If $t_{j,k+1}'\leq t_{j+1,k}$, for $t_{j+1,k}\leq t\leq n$, we have $y_t^{j|(k+1)*} > 0,\, y_t^{j|(k+1)*}\geq y_t^{(j+1)|(k+1)*}$. So $y_t^{j|k*}\geq y_t^{(j+1)|k*}$ due to Lemma \ref{sec3:prop:jgeq}. Because $y_t^{(j+1)|k*}$ has the $(0,1)$ property, $y_t^{(j+1)|k*}=0$ for $1\leq t< t_{j+1,k}$. This is sufficient to show $y_t^{j|k*}\geq y_t^{(j+1)|k*}$ for all $t$.
    
	Otherwise, if $t_{j+1,k}< t_{j,k+1}'$, we just consider the interval $1\leq t < t_{j,k+1}$, because when $t_{j,k+1}\leq t\leq n$, we can using Lemma \ref{sec3:prop:jgeq} to get the conclusion like previous paragraph. Suppose $t_{j,k+1}''$ is the largest $t$ that satisfies $t_{j+1,k}\leq t< t_{j,k+1}'$ and $y_{t-1}^{j|(k+1)*} > 0$. 
	
	From now, consider $t_{j,k+1}''\leq t < t_{j,k+1}'$. In the interval, we have $y_t^{(j+1)|k*}>0$ and $y_t^{j|(k+1)*}=0$. Besides, for $k+1\leq l\leq K$, $y_t^{j|l*} = 0$ as $y_t^{j|l*}\leq y_t^{j|(k+1)*}$ according to Lemma \ref{sec3:prop:iqjk} and $y_t^{(j+1)|l*}=0$ due to hypothesis $y_t^{j|l*}\geq y_t^{(j+1)|l*}$. For $1\leq l\leq k$, we have 
    \begin{eqnarray}
		&&y_t^{(j+1)|l*} + \frac{1}{i_t}\sum_{s=t+1}^{n}\sum_{l=1}^k(y_s^{(j+1)|l*}-y_s^{(j+2)|l*})=\frac{b_{i_t}^l}{nW}, \label{sec3:equ:temp22}
    \end{eqnarray}
	as $y_t^{(j+1)|l*}\geq y_t^{(j+1)|k*} > 0$. We can suppose 
    \begin{eqnarray}
		&&y_t^{j|l} + \frac{1}{i_t}\sum_{s=t+1}^{n}\sum_{l=1}^k(y_s^{j|l*}-y_s^{(j+1)|l*})=\frac{b_{i_t}^l}{nW}, \label{sec3:equ:temp23}
    \end{eqnarray}
	and if we can show $y_t^{j|l}>0$ under this assumption, then, the assumption must be true according to the property of preprocessing part. Let 
    \begin{equation}
        d_j = \sum_{t_{j,k+1}'<s\leq n}\sum_{l=1}^K(y_s^{j|k*}-y_s^{(j+1)|k*})
              = \frac{t_{j,k+1}'b_{t_{j,k+1}'}^k}{nW}-t_{j,k+1}'y_{t_{j,k+1}'}^{j|k*},
        \label{sec3:tempd:1}
    \end{equation} and
    \begin{equation}
        d_{j+1} = \sum_{t_{j,k+1}'<s\leq n}\sum_{l=1}^K(y_s^{(j+1)|k*}-y_s^{(j+2)|k*})
              = \frac{t_{j,k+1}'b_{t_{j,k+1}'}^k}{nW}-t_{j,k+1}'y_{t_{j,k+1}'}^{(j+1)|k*}.
        \label{sec3:tempd:2}
    \end{equation}
		The above Equations \ref{sec3:tempd:1} and \ref{sec3:tempd:2} are obtained due to the fact that the constraints for $y_{t_{j,k+1}'}^{j|k*} > 0$ and $y_{t_{j,k+1}'}^{(j+1)|k*} >0$, that's to say, they satisfy the Equation \ref{sec3:equ:temp16}. We can see $d_j^k \leq d_{j+1}^k$ due to $y_t^{j|k*}\geq y_t^{(j+1)|k*}$ by the induction hypothesis. Apply $d_j$ and $d_{j+1}$ into Equation \ref{sec3:equ:temp22} and \ref{sec3:equ:temp23}, we have 
    \begin{eqnarray}
        &&y_t^{j|l*} + \frac{1}{i_t}\sum_{s=i+1}^{t_{j,k+1}'}\sum_{l=1}^k(y_s^{j|l*}-y_s^{(j+1)|l*}) + \frac{d_j^k}{i_t}=\frac{b_{i_t}^l}{nW}\label{sec3:temp1:1},\\
        &&y_t^{(j+1)|l*} + \frac{1}{i_t}\sum_{s=t+1}^{t_{j,k+1}'}\sum_{l=1}^k(y_s^{(j+1)|l*}-y_s^{(j+2)|l*})+\frac{d_{j+1}^k}{i_t}=\frac{b_{i_t}^l}{nW}\label{sec3:temp1:2},
    \end{eqnarray}
	By adding up $y_t^{j|l*}$ and $y_t^{(j+1)|l*}$ for $l$ from 1 to $k$ from Equation \ref{sec3:temp1:1} and \ref{sec3:temp1:2}, we can obtain
    \begin{eqnarray}
        &&r_{t}^{j|k} + \frac{k}{i_t}\sum_{s=t+1}^{t_{j,k+1}'}(r_s^{j|k}-r_s^{(j+1)|k})+\frac{kd_j^k}{i_t}=\beta_{i_t}^k,\\
        &&r_{t}^{(j+1)|k} + \frac{k}{i_t}\sum_{s=t+1}^{t_{j,k+1}'}(r_s^{(j+1)|k}-r_s^{(j+2)|k})+\frac{kd_{j+1}^k}{i_t}=\beta_{i_t}^k.
    \end{eqnarray}
	Compare above two equations with Lemma \ref{sec3:prop:recurrence}, we can get $r_t^{j|k}\geq r_t^{(j+1)|k}$. Thus $$y_t^{j|k*} = \frac{r_t^{j|k}-\beta_{i_t}^k}{k}+\frac{b_{i_t}^k}{nW}\geq y_t^{(j+1)|k*} = \frac{r_t^{(j+1)|k}-\beta_{i_t}^k}{k}+\frac{b_{i_t}^k}{nW}\ > 0.$$ Then, our assumption is true and we get the result $y_t^{j|k*}\geq y_t^{(j+1)|k*}$ we want.
    
	Next, we show that there doesn't exist such a $t_{j,k+1}''$, which means at least when $t_{j+1,k} \leq t < t_{j,k+1}'$ we have $y_{q,i}^{j|(k+1)*} = 0$. If such $t_{j,k+1}''$ exists, that's $y_{t_{j,k+1}''-1}^{j|(k+1)*} > 0$, we have 
    \begin{equation}
    \begin{split}
		&i_{t_{j,k+1}''-1}y_{t_{j,k+1}''-1}^{j|(k+1)*} \\
		= &\frac{i_{t_{j,k+1}''-1}b_{i_{t_{j,k+1}''-1}}^{k+1}}{nW}-\sum_{s = t_{j,k+1}''}^n\sum_{l=1}^K(y_s^{j|l*}-y_s^{(j+1)|l*})\\
		\leq & \frac{i_{t_{j,k+1}'-1}b_{i_{t_{j,k+1}'-1}}^{k+1}}{nW}-\sum_{s = t_{j,k+1}'}^n\sum_{l=1}^K(y_s^{j|l*}-y_s^{(j+1)|l*})\\
		=& (t_{j,k+1}'-1)y_{t_{j,k+1}'-1}^{j|(k+1)*}\leq  0,
    \end{split}
    \end{equation}
    The third line is obtained from Proposition \ref{sec3:prop:bik} and the fact $y_t^{j|k*}\geq y_t^{(j+1)|k*}$ for $t_{j,k+1}''\leq t\leq n$. It is a contradiction with $y_{t_{j,k+1}''-1}^{j|(k+1)*} > 0$. Up to now, we showed that for  $t_{j+1,i}\leq t\leq n$, $y_t^{j|k*}\geq y_t^{(j+1)|k*}$ is true. When $1\leq t< t_{j+1,k}$ this is necessarily true for $y_t^{(j+1)|k*} = 0$. 
    
	Thus, using induction on $k$, we can show that for $1\leq t\leq n,\ 1\leq k\leq K$, $y_t^{j|k*}\geq y_t^{(j+1)|k*}$. That's to say, the hypothesis 1 for a general $j$ is held. According to the Fact \ref{sec3:fact:temp1}, hypothesis 2 is also held. Then, $y_t^{j|k*}$ has the $(0,1)$ property for $1\leq k\leq K$. We finish the induction part of this lemma for general $j$. 
    
    Finally, using induction on $j$, we finish this proof of Equation \ref{sec3:equ:yt01}. That's the $y_t^{j|k*}$ has the $(0,1)$ property. According to Proposition \ref{sec3:prop:qqueuesto1}, $y_{q,i}^{j|k*}$ also has the $(0,1)$ property and thus finish this lemma. Besides we have
    \begin{equation}
        \iqjk{q}{j}{k} = \left\{
        \begin{aligned}
            &\ttoi{t_{j,k}},& &q \geq ((t_{j,k}-1)\mod Q) +1&\\
            &\floor{\frac{t_{j,k}-1}{Q}}+2,& &q < ((t_{j,k}-1)\mod Q) +1.&
        \end{aligned}\right.
    \end{equation}
    From $y_t^{j|k}\geq y_t^{(j+1)|k}$ for $1\leq t\leq n,\,1\leq j\leq J,\, 1\leq k\leq K$, we can get respective $y_{q,i}^{j|k}\geq y_{q,i}^{(j+1)|k}$ where $t = (i-1)Q+q$. This also means $\iqjk{q}{j}{k} \leq \iqjk{q}{j+1}{k}$.\qed
\end{proof}

 The crucial positions play a key role in the protocol, and up to now, some properties of them have been revealed. We summarize those properties here.
\begin{prop}
    For $1\leq q\leq Q,\,1\leq j\leq J,\, 1\leq k\leq K,$ we have
        $\iqjk{q+1}{j}{k}\leq \iqjk{q}{j}{k},\ 
        \iqjk{q}{j}{k}\leq \iqjk{q}{j+1}{k},$ and  
        $\iqjk{q}{j}{k}\leq \iqjk{q}{j}{k+1}.$

\end{prop}
Employing the complementary slackness theorem, we can show the our protocol is optimal. 
\def\theooptmality{
	Taking the \{$\iqjk{q}{j}{k}\,|\,1\leq q\leq Q,\,1\leq j\leq J,\,1 \leq k\leq K$\} obtained from the preprocessing part as input, the Adaptive Observation-Selection Protocol is optimal for the \swjkq. }
\begin{theo}
	\label{sec3:theo:optimality}
	\theooptmality
\end{theo}

\begin{proof}
	Using $\{x_{q,i}^{j|k*}\,|\,1\leq i\leq n/Q,\,1\leq j\leq J,\,1\leq k\leq K\}$, to stand for the $(0,1)$-solution of the primal program that can be mapped to the Adaptive Observation-Selection Protocol. This means we have
	{\small
	\begin{equation*}
		x_{q,i}^{j|k*}=\left\{
			\begin{aligned}
				& \sum_{m=1}^Q\sum_{s=1}^{i-1}\frac{1}{s}\sum_{l=1}^K\left(x_{m,s}^{j-1|l*}-x_{m,s}^{j|l*}\right)+\sum_{m=1}^{q-1}\frac{1}{i}\sum_{l=1}^K\left(x_{m,i}^{j-1|l*}-x_{m,i}^{j|l*}\right)>0,&
				&i_{q,j,k}\leq i\leq \frac{n}{Q}&\\
				&  0,&
				&1\leq i < i_{q,j,k}.&
			\end{aligned}
			\right.
		\end{equation*}
	}
	for $1\leq q\leq Q,\ 1\leq j\leq J,\ 1\leq k\leq K$. Note that in above equation, we use a set of dummy values $x_{q,i}^{0|k}$ for convenience as mentioned in the definition of the primal program.
		On the other hand, we have 
	{\small
\begin{equation*}
y_{q,i}^{j|k*}=\left\{
	\begin{aligned}
	&\frac{b_i^k}{nW}-\frac{1}{i}\sum_{m=1}^Q\sum_{s=i+1}^{n/Q}\sum_{l=1}^K\left(y_{m,s}^{j|l*}-y_{m,s}^{(j+1)|l*}\right)-\frac{1}{i}\sum_{m=q+1}^Q\sum_{l=1}^K\left(y_{m,i}^{j|l*}-y_{m,i}^{(j+1)|l*}\right) > 0,&\\
		&(\iqjk{q}{j}{k}\leq i\leq n/Q)&\\
		&0, ~~(1\leq i < \iqjk{q}{j}{k}).&
	\end{aligned}\right.
\end{equation*}
}
for $1\leq q\leq Q,\ 1\leq j\leq J,\ 1\leq k\leq K$.

	Using $xs_{q,i}^{j|k*}$ and $ys_{q,i}^{j|k*}$ to stand for the value of slackness variables of $x_{q,i}^{j|k*}$ and $y_{q,i}^{j|k*}$. Then we have 
				$x_{q,i}^{j|k*}\cdot ys_{q,i}^{j|k*} = 0$ and
				$y_{q,i}^{j|k*}\cdot xs_{q,i}^{j|k*} = 0,$
for all $q,i,j,k$.
			This is because when $1\leq i < \iqjk{q}{j}{k}$, both $x_{q,i}^{j|k*}$ and $y_{q,i}^{j|k*}$ equal to 0; when $\iqjk{q}{j}{k}\leq i\leq n$, both $qs_{q,i}^{j|k*}$ and $ys_{q,i}^{j|k*}$ equal to 0 due to the crucial constraints of $x_{q,i}^{j|k}$ and $y_{q,i}^{j|k}$ are tight. Through the theorem of complementary slackness, we know the two $(0,1)$-solutions are optimal for their respective program. Thus, the Adaptive Observation-Selection Protocol is optimal for the \swjkq.\qed
\end{proof}

\vspace{-0.4cm}

\newcommand{\comp}[3]{\alpha(#1,#2,#3)}
\vspace{-0.2cm}
\section{Extensions and Analysis of the Optimal Protocol\label{sec:name:analysis1}}
\vspace{-0.2cm}

\subsection{Applications in Other Generalizations\label{sec:name:application}}
\vspace{-0.2cm}

Our optimal protocol is based on the essential structure of the LP model. Several variants can be characterized by LP model with similar structure. Thus our optimal protocol can be extended to solve these related variants.

It is obvious that we can obtain an optimal protocol for weighted $J$-choice $K$-best secretary problem when $Q$ is set to be 1.
Based on the $J$-choice $K$-best problem, we consider another variant: the employer just interviews the first $m$ candidates, $1\leq m\leq n$, due to time or resource limitation. Other settings keep unchanged. We call this problem \emph{fractional $J$-choice $K$-best secretary problem}. We can characterize this problem by a LP program called  $FLP$ as follow: 
\begin{equation*}
	\label{sec:app:flp}
    \begin{split}
		FLP: &\qquad\max z=\frac{1} {nW}\sum_{j=1}^J\sum_{l=1}^K\sum_{i=1}^m\sum_{k=l}^Kw_k \frac{\binom{i-1}{l-1}\binom{n-i}{k-l}}{\binom{n-1}{k-1}}x_i^{j|k}\\
        &\text{s.t.}\left\{
        \begin{aligned}
			&x_i^{j|k}\leq \sum_{s=1}^{i-1}\frac{1}{s}\sum_{l=1}^K\left(x_s^{(j-1)|l}-x_s^{j|l}\right),& \\
			&(1\leq i\leq m,\,1\leq k\leq K,\,1\leq j\leq J)& \\
   &x_i^{j|k}\geq 0, ~~ (1\leq i\leq m,\,1\leq k\leq K,\,1\leq j\leq J). &
        \end{aligned}
        \right.\\
    \end{split}
\end{equation*}
Note that, like the LP \ref{sec3:equ:lpprimal}, we add some dummy variables $x_i^{0|k}$, $1\leq i\leq m$ and $1\leq k\leq K$, and set $\sum_{s=1}^{i-1}\frac{1}{s}\sum_{l=1}^Kx_s^{0|l} = 1$ so that the constraints of this $FLP$ has a uniform form.

The $FLP$ has the same structure with the LP \ref{sec3:equ:lpprimal}, and all the properties used to show the optimality of the Adaptive Observation-Selection protocol are still held. Thus, our protocol can be easily generalized to solve this problem.

In the \swjkq, all interviewers \emph{share} the $J$ quotas.
Another case is that a fixed quota is preallocated to each queue, that's to say, in any queue $q$, the employer can only hire at most $J_q$ candidates where $J = \sum_{q=1}^Q J_q$. Besides, we suppose there are $n_q$ candidates in queue $q$ so that $n = \sum_{q=1}^Q n_q$. Other settings, except the synchronous requirement, keep unchanged compared to the \swjkq. This is the problem which is called \emph{\fewjkq}\ (abbreviated as \ewjkq). Feldman et al. \cite{Feldman:2012} have considered the non-weighted version of the \ewjkq\ with the condition $J=K$. Actually, for each queue of the \ewjkq, since what we care about is the expectation and the candidates' information and quotas can not be shared, how employer selects candidate has \emph{no influence} on other queues. So, it is an independent fractional weighted $J_q$-choice $K$-best secretary problem with $m=n_q$ in each queue. Then, running the modified Adaptive Observation-Selection protocol on each queue is an optimal protocol for \ewjkq.

\vspace{-0.2cm}
\subsection{Competitive Ratio Analysis \label{sec:name:ratioanalysis}}
\vspace{-0.1cm}
Let $\comp{Q}{J}{K}$ stand for the competitive ratio of Adaptive Observation-Selection Protocol. For the general case, $\comp{Q}{J}{K}$ is complicated to analyze either from the view of protocol or the dual program.
In this section, we provide analysis about two typical cases: the $(1,1,K)$ case and the $(2,2,2)$ case. Both the cases we deal with are the uniformly weighted (or non-weighted) versions of \swjkq, i.e. $w_1=w_2=\cdots=w_K=1$.

The first one we study is the $(1,1,K)$ case that selecting $1$ candidate among the top $K$ of $n$ candidates with just one queue. It is also called $K$-best problem.
Suppose $\gamma_1$ and $\gamma_2$ are real numbers that satisfy $0\leq \gamma_1\leq \gamma_2\leq 1$. Consider the Algorithm \ref{sec4:alg:simp}.

\begin{algorithm}
	\label{sec4:alg:simp}
	\caption{Simple Algorithm for $(1,1,K)$ Problem}
	\Input {$n$, $K$}
	\Output {the candidate selected}
	{\bf for} the first $\floor{\gamma_1 n} -1$ candidates, just interview but don't select anyone\\
	\For{$i=\floor{\gamma_1 n}$ to $\floor{\gamma_2 n}-1$}
	{
		\If{the $i$-th candidate is better than anyone previous seen}
		{
			select this candidate and exit
		}
	}
	\For{$i=\floor{\gamma_2 n}$ to $n$}
	{
		\If{the $i$-th candidate is the best or second best candidate up to now}
		{
			select the $i$-th candidate and exit
		}
	}
\end{algorithm}
As our Adaptive Observation-Selection protocol is optimal, the performance of this Algorithm \ref{sec4:alg:simp} is a lower bound of our protocol. We get the following lower bound of $\alpha(1,1,K)$ based on the analysis of this three-phase algorithm. We have the following theorem.

\def\theoratiok{
		$\comp{1}{1}{K} \geq 1-O\left(\frac{\ln^2K}{K^2}\right)$ when $K$ is large enough and $n\gg K$.
	}
	\vspace{-0.1cm}
	\begin{theo}
		\label{sec4:theo:ratiok}
		\theoratiok
	\end{theo}

\begin{proof}
	For the concision of the proof, we suppose both $\gamma_1n$ and $\gamma_2n$ are integers without loss of generality.

	Define the range from $\gamma_1 n$-th candidate to $(\gamma_2n-1)$-th candidate as \emph{Phase 1} and the range from $\gamma_2n$-th candidate to $n$-th candidate as \emph{Phase 2}.
	Let $T_k$, $1\leq k\leq K$, stand for the $k$-th best candidate and $A_k$ stand for the event that the $T_k$ is selected by the algorithm. More specifically, denote  $A_{k,j}^l$, $j,l\in\{0,1\}$, as event that $T_k$ is selected in Phase $j$ when he/she is the $l$-th best up to now. 

	Suppose $T_k$ is selected is Phase $1$. His/her position is $i$ with probability $\frac{1}{n}$. $T_k$ must be the \emph{best} candidate for $1$ to $i$. That's to say the best candidate from position $1$ to $i-1$ must come before $\gamma_1n$, which happens with probability $\frac{\gamma_1n-1}{i-1}$. Besides all candidates that better than $T_k$ must come after $i$. The probability of this event is $\binom{n-i}{k-1}/\binom{n-1}{k-1}$. To sum up, we have 
\begin{equation*}
	\Pr(A_{k,1}^1)= \sum_{i = \gamma_1 n}^{\gamma_1 n-1}\frac{1}{n}\cdot \frac{\gamma n-1}{i-1}\cdot \frac{\binom{n-i}{k-1}}{\binom{n-1}{k-1}}\\
\end{equation*}

When selected in Phase 2, the $T_k$ can be the best or the second best up to now. Then, similar to in Phase 1, $$\Pr(A_{k,2}^1) = \sum_{i = \gamma_1 n}^{\gamma_1 n-1}\frac{1}{n}\cdot \frac{(\gamma_1 n-1)(\gamma_2 n-2)}{(i-1)(i-2)}\cdot \frac{\binom{n-i}{k-1}}{\binom{n-1}{k-1}}.$$ $(\gamma_2n-2)/(i-2)$ in above formula means the probability that the second best candidate from position $1$ to $i-1$ must come before $\gamma_2n$.

When $A_{k,2}^2$ happens, it means there is exact one candidate that better than $T_k$ comes before position $\gamma_1n$ and the second best candidate from 1 to $i-1$ comes before $\gamma_2n-1$. So, there has 
	\begin{equation*}
		\begin{aligned}
			&\Pr(A_{k,2}^2) &=&\frac{1}{n}\sum_{i=\gamma_2n}^n&& \frac{\binom{n-k}{i-2}(k-1)(\gamma_1n-1)(\gamma_2n-2)(i-3)!(n-i)!}{(n-1)!}&\\
					 &&=&\frac{1}{n}\sum_{i=\gamma_2n}^n&& \frac{(k-1)(\gamma_1n-1)(\gamma_2n-2)\binom{n-i}{k-2}}{(n-1)(i-2)\binom{n-2}{k-2}}.&
		\end{aligned}
	\end{equation*} 
	In the first line of above equation, $\binom{n-k}{i-2}(k-1)$ means all possible ways to choose $i-1$ candidates that there is exact $1$ candidate better than $T_k$.

	We define the ratio of Algorithm \ref{sec4:alg:simp} as $\alpha_{K,2}$. Then we have $$\alpha_{K,2} = \sum_{k=1}^K \Pr(A_k) = \sum_{k=1}^K\left(\Pr(A_{k,1}^1)+\Pr(A_{k,2}^1)+\Pr(A_{k,2}^2)\right).$$ We calculate its value separately as follow.
{\small
		\begin{equation*}
			\begin{split}
				\sum_{k=1}^K\Pr(A_{k,1}^1)
				=&\frac{\gamma_1n-1}{n}\sum_{i = \gamma_1n}^{\gamma_2n-1}\frac{1}{(i-1)\binom{n-1}{i-1}}\sum_{k=1}^K\binom{n-k}{i-1}\\
				=&\frac{\gamma_1n-1}{n}\sum_{i = \gamma_1n}^{\gamma_2n-1}\frac{\binom{n}{i}-\binom{n-K}{i}}{(i-1)\binom{n-1}{i-1}}\\
				\geq &\left(1-\frac{\gamma_1n-1}{\gamma_2n-1}\right)\left(1-\frac{(n-\gamma_1n)\cdots(n-\gamma_1n-K+1)}{n(n-1)\cdots(n-K+1)}\right)
			\end{split}
		\end{equation*}
		
		\begin{equation*}
			\begin{split}
				\sum_{k= 1}^K \Pr(A_{k,2}^1)
				=&\frac{(\gamma_1n-1)(\gamma_2n-2)}{n}\sum_{i=\gamma_2n}^{n}\frac{\binom{n}{i}-\binom{n-K}{i}}{(i-1)(i-2)\binom{n-i}{i-1}}\\
				=&\frac{(\gamma_1n-1)(\gamma_2n-2)}{n}\sum_{i=\gamma_2n}^{n}\left(\frac{n}{i(i-1)(i-2)}-\frac{\binom{n-K}{i}}{(i-1)(i-2)\binom{n-i}{i-1}}\right)\\
				\geq &\frac{1}{2}\left(\frac{\gamma_1n-1}{\gamma_2n-1}-\frac{(\gamma_1n-1)(\gamma_2n-2)}{n(n-1)}\right)\left(1-\frac{(n-\gamma_2n)\cdots(n-\gamma_2n-K+1)}{n(n-1)\cdots(n-K+1)}\right)
			\end{split}
		\end{equation*}

		\begin{equation*}
			\begin{split}
				\sum_{k=1}^K\Pr(A_{k,2}^2)
				=&\frac{(\gamma_1n-1)(\gamma_2n-2)}{n(n-1)}\sum_{i=\gamma_2n}^n\frac{1}{(i-2)\binom{n-2}{i-2}}\sum_{k=1}^K(k-1)\binom{n-k}{i-2}\\
				=&\frac{(\gamma_1n-1)(\gamma_2n-2)}{n(n-1)}\sum_{i=\gamma_2n}^n\frac{i-1}{(i-2)\binom{n-2}{i-2}}\sum_{k=1}^K\left(\frac{n-i+1}{i-1}-\frac{n-k-i+2}{i-1}\right)\binom{n-k}{i-2}\\
				=&\frac{(\gamma_1n-1)(\gamma_2n-2)}{n(n-1)}\sum_{i=\gamma_2n}^n\frac{i-1}{(i-2)\binom{n-2}{i-2}}\sum_{k=1}^K\left(\frac{n-i+1}{i-1}\binom{n-k}{i-2}-\binom{n-k}{i-1}\right)\\
				=&\frac{(\gamma_1n-1)(\gamma_2n-2)}{n(n-1)}\sum_{i=\gamma_2n}^n\frac{1}{(i-2)\binom{n-2}{i-2}}\left(\binom{n}{i}-K\binom{n-K}{i-1}-\binom{n-K}{i}\right)\\
				\geq &\frac{1}{2}\left(\frac{\gamma_1n-1}{\gamma_2n-1}-\frac{(\gamma_1n-1)(\gamma_2n-2)}{n(n-1)}\right)\left(1-\frac{(n-\gamma_2n)\cdots(n-\gamma_2n-K+1)}{n(n-1)\cdots(n-K+1)}\right)\\
																	 &\qquad -\frac{K(\gamma_1n-1)(\gamma_2n-2)(n-\gamma_2n+1)(n-\gamma_2n)\cdots(n-K-\gamma_2n+2)}{n(n-1)\cdots(n-K+1)(n-1)(\gamma_2n-2)}.
			\end{split}
		\end{equation*}
}%
Then the $\alpha_{K,2}$ can be estimated. When $n$ is large enough and $n\gg k$, we have
		\begin{equation*}
			\begin{split}
				\lim_{n\rightarrow\infty}\alpha_{K,2} =& \lim_{n\rightarrow\infty}\left(\sum_{k=1}^K\left(\Pr(A_{k,1})+\Pr(A_{k,2}^1)+\Pr(A_{k,2}^2)\right)\right)\\
				\geq &\left(1-\frac{\gamma_1}{\gamma_2}\right)\left(1-(1-\gamma_1)^K\right)-\left(\frac{\gamma_1}{\gamma_2}-\gamma_1\gamma_2+K\gamma_1\right)(1-\gamma_2)^K+\frac{\gamma_1}{\gamma_2}-\gamma_1\gamma_2.
			\end{split}
		\end{equation*}
Define $x= K^{-\frac{2}{K}}$, and let $\gamma_1 = 1-x$, and $\gamma_2 = 1-x^2$, we get 
		\begin{equation*}
			\lim_{n\rightarrow\infty}\alpha_{K,2} \geq 1-(1-x)(1-x^2)-\frac{x^{K+1}}{x+1}+\frac{(x^4+(K-2)x^2-K)x^{2K}}{x+1}.
		\end{equation*}
On the other hand,
		\begin{equation}
			\begin{split}
				x &= e^{-\frac{2\ln(K)}{K}} = 1-\frac{2\ln(K)}{K}+2\left(\frac{\ln(K)}{K}\right)^2 + o\left(\frac{\ln^2(K)}{K^2}\right)\\
				x^2&= e^{-\frac{4\ln(K)}{K}} = 1-\frac{4\ln(K)}{K}+8\left(\frac{\ln(K)}{K}\right)^2 + o\left(\frac{\ln^2(K)}{K^2}\right).\\
			\end{split}
		\end{equation}
		Thus, we can conclude that $\alpha_{K,2} \geq 1-O\left(\frac{\ln^2K}{K^2}\right)$ when $n,K$ are large enough and $n\gg K$. Finally, we have $\alpha(1,1,K)\geq \alpha_{K,2} \geq 1-O\left(\frac{\ln^2K}{K^2}\right)$. \qed
\end{proof}

The Adaptive Observation-Selection protocol performs much better in fact.
Table \ref{sec4:table:1} is the result of numerical experiment for small $K$.
As we can see, $\alpha(1,1,K)$ goes to 1 sharply. 
But it is too complex to analyze when there are $K+1$ phases.

\vspace{0.1cm}
\begin{table}
	\begin{center}
		{\small
		\begin{tabular}{|c|c|c|c|c|c|c|c|c|c|}
			\hline
			$K=1$&$K=2$&$K=3$&$K=4$&$K=5$&$K=6$&$K=7$&$K=8$&$K=9$&$K=10$\\
			0.3679 & 0.5736 & 0.7083 &0.7988 &0.8604 & 0.9028 & 0.9321 &0.9525 & 0.9667 & 0.9766\\\hline
			$K=11$&$K=12$&$K=13$&$K=14$&$K=15$&$K=16$&$K=17$&$K=18$&$K=19$&$K=20$\\
			0.9835 & 0.9884 & 0.9918 & 0.9942 & 0.9959 &0.9971 & 0.9980 & 0.9986 & 0.9990 & 0.9993\\\hline
			$K=21$&$K=22$&$K=23$&$K=24$&$K=25$&$K=26$&$K=27$&$K=28$&$K=29$&$K=30$\\
			0.9995 & 0.9996 & 0.9997 & 0.9998 & 0.9999 &0.9999 & 0.9999 & $>$0.9999 & $>$0.9999 & $>$0.9999\\
			\hline
		\end{tabular}
	}
	\end{center}
	\caption{the value of $\alpha(1,1,K)$ when $n=10000$\label{sec4:table:1}}
\end{table}

Another case is when $Q=J=K=2$. The main idea is to calculate the optimal $(0,1)$-solution of the dual program based on the preprocessing part.
This analysis is almost accurate when $n$ is large enough.
We have the following result.
\def\theoforttt{
	When $n$ is large enough, the Adaptive Observation-Selection protocol achieves a competitive ratio $\comp{2}{2}{2} \approx 0.372$.
}
\begin{theo}
	\label{sec4:theo:theoforttt}
	\theoforttt
\end{theo}

\begin{proof}
	The main idea of the proof is calculate the optimal solution of the dual program according to preprocessing part. The method is based on the proof of Lemma \ref{sec3:lemm:01ofdual}.

	In this proof, we employ a set of real numbers $\{\gamma_{j,k}\,|\, 1\leq k\leq 2,\, 1\leq j\leq 2\}$ that satisfy $\lim_{n\rightarrow \infty}\frac{\iqjk{q}{j}{k}}{n} = \gamma_{j,k}$. Note that $\gamma_{j,k}$ is independent on $q$ because $|\iqjk{1}{j}{j}-\iqjk{2}{j}{k}|\leq 1$. As what we concern is the value of $\gamma_{j,k}$, we can consider that $\iqjk{1}{j}{j}$ is equal to $\iqjk{2}{j}{k}$ in the following proof as $n$ is large enough.
	We define $r_{q,i}^{j|k} = \sum_{l = 1}^k y_{q,i}^{j,l}$ and $R_{i,j} = \sum_{s = i}^{n/2}\sum_{q=1}^Q r_{q,s}^{j|K}$. Without loss of generality, we suppose $n$ is even. 

When the $\ittt\leq i\leq n/2$, we know 
\begin{equation}
	\begin{split}
		&y_{2,i}^{2|k}+\frac{1}{i}\sum_{s=i+1}^{n/2}\sum_{m=1}^2\sum_{l=1}^2y_{m,s}^{2|l} = \frac{b_i^k}{2n},\\
		&y_{1,i}^{2|k}+\frac{1}{i}\sum_{l=1}^2\left(\sum_{s=i+1}^{n/2}\sum_{m=1}^2y_{m,s}^{1|l}+y_{2,i}^{2|l}\right) = \frac{b_i^k}{2n}.\\
	\end{split}
\end{equation}
For above equations, add up $k=1,2$, we get
\begin{equation}
	\label{sec:app:temp37}
	\begin{split}
		&r_{2,i}^{2|2}+\frac{2}{i}\sum_{s=i+1}^{n/2}\sum_{m=1}^2r_{m,s}^{2|2} = \frac{1}{n},\\
		&r_{1,i}^{2|2}+\frac{2}{i}\left(\sum_{s=i+1}^{n/2}\sum_{m=1}^2r_{m,s}^{2|2}+r_{2,i}^{2|2}\right) = \frac{1}{n}.\\
	\end{split}
\end{equation}
Applying $R_{i,2}$ to above equation, we can get
\begin{equation}
	\begin{split}
		&r_{2,i}^{2|2}+\frac{2}{i}R_{i+1,2} = \frac{1}{n},\\
		&r_{1,i}^{2|2}+\frac{2}{i}\left(r_{2,i}^{2|2}+R_{i+1,2}\right)=\frac{1}{n}.
	\end{split}
\end{equation}
On the other hand, we have $R_{i,2} = r_{1,i}^{2|2}+r_{1,i}^{2|2}+R_{i+1,2}$. Thus we can easily get the follow recursion about $R_{i,2}$.
\begin{equation}
	R_{i,2} = \left\{\begin{aligned}
			&0,& & i = n/2 +1&\\
			&\frac{(i-2)^2}{i^2}R_{i+1,2}+\frac{2}{n}\left(1-\frac{1}{i}\right),& &\ittt \leq i\leq n/2.&\\
	\end{aligned}
	\right.
\end{equation}
Solving this recursion we have $R_{i,2} = \sum_{l=i}^{n/2}\frac{(i-1)^2(i-2)^2}{(l-1)^2(l-2)^2}\frac{2}{n}\left(1-\frac{1}{i}\right)$. When $n\rightarrow \infty$, $R_{i,2} = \frac{2i}{3n}-\frac{16i^4}{3n^4}=\frac{2\gamma_{2,2}}{3}-\frac{16\gamma^4}{3}$.

Now we want to know the value of $\gamma_{2,2}$. When the constraint is tight we have $$y_{2,i}^{2|2}+\frac{2}{i}R_{i+1,2} = \frac{i-1}{2n(n-1)}.$$ Considering the property of $\ittt$. There are 
\begin{equation}
	\begin{split}
		&y_{2,\ittt}^{2|2} = \frac{\ittt-1}{2n(n-1)}-\frac{1}{\ittt}R_{\ittt+1, 2} \geq 0,\\
	 &y_{2,\iqjk{2}{2}{2}-1}^{2|2} = \frac{\ittt-2}{2n(n-1)}-\frac{1}{\ittt-1}R_{\ittt, 2} \leq 0.\\
	\end{split}
\end{equation}
When $n\rightarrow$, we can consider $\frac{\gamma_{2,2}^2}{2}-2R_{i,2} = 0,$ without loss much of accuracy of $\gamma_{2,2}$. Then, we have $$\frac{1}{2}\gamma_{2,2}^2-\frac{2}{3}\gamma_{2,2}+\frac{16}{3}\gamma_{2,2}^4 = 0.$$
Solving above equation we get $\gamma_{2,2} \approx 0.4379$.

When $\itto\leq i \leq \ittt -1$, we know $y_{2,i}^{2|2} = 0$, thus $r_{2,i}^{2|2}= r_{2,i}^{2|1} = y_{2,i}^{2|1}$. Similar to Equation \ref{sec:app:temp37}, following equation can be obtained 
\begin{equation}
	\begin{split}
		&r_{2,i}^{2|2}+\frac{1}{i}\sum_{s=i+1}^{n/2}\sum_{m=1}^2r_{m,s}^{2|2} = \frac{2n-i-1}{2n(n-1)},\\
		&r_{1,i}^{2|2}+\frac{1}{i}\left(\sum_{s=i+1}^{n/2}\sum_{m=1}^2r_{m,s}^{2|2}+r_{2,i}^{2|2}\right) = \frac{2n-i-1}{2n(n-1)}.\\
	\end{split}
\end{equation}
Similarly, the following recursion is held in this interval
\begin{equation}
	R_{i, 2}=\left\{
		\begin{aligned}
			&R_{\ittt,2},&& i = \ittt&\\
			&\frac{(i-1)^2}{i^2}R_{i+1,2}+\left(\frac{2n-i-1}{2n(n-1)}\right)\left(2-\frac{1}{i}\right), && \itto\leq i < \ittt.&
		\end{aligned}\right.
\end{equation}
When $n \rightarrow \infty$, solving this recursion we get $R_{i,2}=\frac{i^2}{2n^2}+\frac{2i}{n}-\frac{2i^2}{\gtt n^2}-\frac{i^2}{n^2}\ln(\frac{\gtt n}{i})$. Again, solving $R_{i+1,2}=\frac{i(2n-i-1)}{2n(n-1)}$ we can get the value of $\gto$ accurate enough. That the $\gto$ satisfies $$\frac{1}{2}\gto^2+2\gto-\frac{2}{\gtt}\gto^2-\gto^2\ln\left(\frac{\gtt}{\gto}\right)=\gto\left(1-\frac{\gto}{2}\right).$$ We can get $\gto \approx 0.2398$. 

	The procedure to calculate $\goo$ and $\got$ is the same but more complex and tedious. We simply list the main result here.
{\small 
	\begin{equation*}
		R_{i,1} =\left\{
		\begin{aligned}
			&\frac{14i}{9n}-\frac{112i}{9n}-\frac{64i^4}{3n^4}\ln(\frac{n}{2i}), &\\
			&(\ittt\leq i \leq n/2)&\\
			&\frac{10i}{3n}+\left(1-\ln\left(\frac{\ittt}{i}\right)-\frac{2}{\gtt}\right)\frac{2i^2}{n^2}+\left(\frac{2}{3\gtt^3}-\frac{2}{\gtt^2}+\frac{R_{\ittt,1}}{\gtt^4}\right)\frac{i^4}{n^4},&\\
			&(\itot\leq i < \ittt)&\\
			&\frac{6i}{n}-\left(\frac{6}{\got }+\frac{4}{\gtt}\ln\left(\frac{\got n}{i}\right)+\ln\left(\frac{\got n}{i}\right)\ln\left(\frac{\gtt^2 n}{\got i}\right)-\frac{R_{\itot,1}}{\got^2}\right)\frac{i^2}{n^2},&\\
			&(\itto \leq i < \itot)&\\
			&\frac{2i}{n}-\frac{2i^2}{\gto n^2}-\frac{i^2}{n^2}\ln\left(\frac{\gto n}{i}\right)+R_{\itto,2}\left(1-\frac{i^2}{\gto^2n^2}\right)+\frac{R_{\itto,1}i^2}{\gto^2n^2},&\\
			&(\itoo \leq i < \itto).&\\
		\end{aligned}
		\right.
	\end{equation*}
}
Besides, $\goo,\, \got$ satisfy
{\small
\begin{equation*}
	\begin{split}
		&\frac{10\got}{3}+\left(1-\ln\left(\frac{\gtt}{\got}\right)-\frac{2}{\gtt}\right)2\got^2+\left(\frac{2}{3\gtt^3}-\frac{2}{\gtt^2}+\frac{R_{\ittt,1}}{\gtt^4}\right)\got^4 =\frac{\got^2}{2}+R_{\itot,2},\\
		&2\goo-\frac{2\goo^2}{\gto}-\goo^2\ln\left(\frac{\gto }{\goo}\right)+R_{\itto,2}\left(1-\frac{\goo^2}{\gto^2}\right)+\frac{R_{\itto,1}\goo^2}{\gto^2}=\goo\left(1-\frac{\goo}{2}\right)+R_{\itto, 2}.
	\end{split}
\end{equation*}
}
Finally, we get $\goo \approx 0.1765,\, \got \approx 0.3658$ and $\comp{2}{2}{2} = R_{\iooo,1} \approx 0.372$.\qed
\end{proof}

\vspace{-0.3cm}
\section{Conclusion\label{sec:name:conclusion}}
\vspace{-0.2cm}
In this paper, we deal with a generalization of secretary problem in the parallel setting, the \fswjkq, and provide a deterministic optimal protocol. This protocol can be applied to a series of relevant variants while keeps optimal. In addition, we provide some analytical results for two typical cases: the 1-queue 1-choice $K$-best case and the shared 2-queue 2-choice 2-best case. 

There are several interesting open problems. The first one is making a tighter analysis of the competitive ratio for \fswjkq. For the $1$-queue $1$-choice $K$-best case, we conjecture that the competitive ratio has the form of $1-O(f(K)^K)$ for some negligible function $f$. For the general case, there is no notable result up to now and lots of work remain to be done. Another interesting aspect is 
to know whether the technique in this paper can be used to find deterministic protocol for other variations such as matroid secretary problem, submodular secretary problem, knapsack secretary problem etc.


\vspace{-0.2cm}
\begin{thebibliography}{10}

\bibitem{Ajtai:2001}
Miklos Ajtai, Nimrod Megiddo, and Orli Waarts.
\newblock Improved algorithms and analysis for secretary problems and
  generalizations.
\newblock {\em SIAM Journal on Discrete Mathematics}, 14(1):1--27, 2001.

\bibitem{Babaioff:2007}
Moshe Babaioff, Nicole Immorlica, David Kempe, and Robert Kleinberg.
\newblock A knapsack secretary problem with applications.
\newblock {\em Proceedings of the 10th International Workshop on Approximation
  Algorithms for Combinatorial Optimization Problems}, pages 16--28, 2007.

\bibitem{Babaioff:2008}
Moshe Babaioff, Nicole Immorlica, David Kempe, and Robert Kleinberg.
\newblock Online auctions and generalized secretary problems.
\newblock {\em SIGecom Exchanges}, 7(2):7:1--7:11, June 2008.

\bibitem{Babaioff:2007:Mar}
Moshe Babaioff, Nicole Immorlica, and Robert Kleinberg.
\newblock Matroids, secretary problems, and online mechanisms.
\newblock In {\em Proceedings of the 18th annual ACM-SIAM Symposium on Discrete
  Algorithms}, pages 434--443, 2007.

\bibitem{Bateni:2010}
Mohammad Bateni, Mohammad Hajiaghayi, and Morteza Zadimoghaddam.
\newblock Submodular secretary problem and extensions.
\newblock {\em Proceedings of the 13th International Workshop on Approximation
  Algorithms for Combinatorial Optimization Problems}, pages 39--52, 2010.

\bibitem{Borosan:2009}
Peter Borosan and Mudassir Shabbir.
\newblock A survey of secretary problem and its extensions, 2009.

\bibitem{Buchbinder:2010}
Niv Buchbinder, Kamal Jain, and Mohit Singh.
\newblock Secretary problems via linear programming.
\newblock In {\em Proceedings of the 14th international conference on Integer
  Programming and Combinatorial Optimization}, pages 163--176, 2010.

\bibitem{Chakraborty:2012}
Sourav Chakraborty and Oded Lachish.
\newblock Improved competitive ratio for the matroid secretary problem.
\newblock In {\em Proceedings of the 23rd Annual ACM-SIAM Symposium on Discrete
  Algorithms}, pages 1702--1712, 2012.

\bibitem{Chan:2013}
T.-H.~Hubert Chan and Fei Chen.
\newblock A primal-dual continuous lp method on the multi-choice multi-best
  secretary problem.
\newblock {\em CoRR}, abs/1307.0624, 2013.

\bibitem{Dimitrov:2008}
Nedialko~B. Dimitrov and C.~Greg Plaxton.
\newblock Competitive weighted matching in transversal matroids.
\newblock In {\em Proceedings of the 35th international colloquium on Automata,
  Languages and Programming, Part I}, pages 397--408, 2008.

\bibitem{Dinitz:2013}
Michael Dinitz and Guy Kortsarz.
\newblock Matroid secretary for regular and decomposable matroids.
\newblock In {\em Proceedings of the 24th annual ACM-SIAM Symposium on Discrete
  Algorithms}, pages 108--117, 2013.

\bibitem{Dynkin:1963}
Eugene~B Dynkin.
\newblock The optimum choice of the instant for stopping a markov process.
\newblock {\em soviet Mathematics Doklady}, 4, 1963.

\bibitem{Feldman:2011}
Moran Feldman, Joseph~Seffi Naor, and Roy Schwartz.
\newblock Improved competitive ratios for submodular secretary problems.
\newblock {\em Proceedings of the 14th International Workshop on Approximation
  Algorithms for Combinatorial Optimization Problems}, pages 218--229, 2011.

\bibitem{Feldman:2012}
Moran Feldman and Moshe Tennenholtz.
\newblock Interviewing secretaries in parallel.
\newblock In {\em Proceedings of the 13th ACM Conference on Electronic
  Commerce}, pages 550--567, 2012.

\bibitem{Freeman:1983}
PR~Freeman.
\newblock The secretary problem and its extensions: A review.
\newblock {\em International Statistical Review}, pages 189--206, 1983.

\bibitem{Gardner:1960}
Martin Gardner.
\newblock Mathematical games.
\newblock {\em Scientific American}, pages 150--153, 1960.

\bibitem{Georgiou:2008}
Nicholas Georgiou, Malgorzata Kuchta, Michal Morayne, and Jaroslaw Niemiec.
\newblock On a universal best choice algorithm for partially ordered sets.
\newblock {\em Random Structures and Algorithms}, 32(3):263--273, 2008.

\bibitem{Gharan:2013}
Shayan~Oveis Gharan and Jan Vondr{\'a}k.
\newblock On variants of the matroid secretary problem.
\newblock {\em Algorithmica}, 67(4):472--497, 2013.

\bibitem{Gupta:2010}
Anupam Gupta, Aaron Roth, Grant Schoenebeck, and Kunal Talwar.
\newblock Constrained non-monotone submodular maximization: offline and
  secretary algorithms.
\newblock In {\em Proceedings of the 6th International Conference on Internet
  and Network Economics}, pages 246--257, 2010.

\bibitem{Hajiaghayi:2004}
Mohammad Hajiaghayi, Robert Kleinberg, and David~C. Parkes.
\newblock Adaptive limited-supply online auctions.
\newblock In {\em Proceedings of the 5th ACM conference on Electronic
  Commerce}, pages 71--80, 2004.

\bibitem{Im:2011}
Sungjin Im and Yajun Wang.
\newblock Secretary problems: laminar matroid and interval scheduling.
\newblock In {\em Proceedings of the 22nd Annual ACM-SIAM Symposium on Discrete
  Algorithms}, pages 1265--1274, 2011.

\bibitem{Jaillet:2013}
Patrick Jaillet, Jos{\'e}~A Soto, and Rico Zenklusen.
\newblock Advances on matroid secretary problems: Free order model and laminar
  case.
\newblock In {\em Proceedings of the 16th international conference on Integer
  Programming and Combinatorial Optimization}, pages 254--265, 2013.

\bibitem{Kesselheim:2013}
Thomas Kesselheim, Klaus Radke, Andreas T{\"o}nnis, and Berthold V{\"o}cking.
\newblock An optimal online algorithm for weighted bipartite matching and
  extensions to combinatorial auctions.
\newblock In {\em Proceedings of the 21st European Symposium on Algorithms},
  pages 589--600, 2013.

\bibitem{Kleinberg:2005}
Robert Kleinberg.
\newblock A multiple-choice secretary algorithm with applications to online
  auctions.
\newblock In {\em Proceedings of the 16th annual ACM-SIAM Symposium on Discrete
  Algorithms}, pages 630--631, 2005.

\bibitem{Korula:2009}
Nitish Korula and Martin P\'{a}l.
\newblock Algorithms for secretary problems on graphs and hypergraphs.
\newblock In {\em Proceedings of the 36th Internatilonal Collogquium on
  Automata, Languages and Programming: Part II}, pages 508--520, 2009.

\bibitem{Koutsoupias:2013}
Elias Koutsoupias and George Pierrakos.
\newblock On the competitive ratio of online sampling auctions.
\newblock {\em ACM Transactions on Economics and Computation},
  1(2):10:1--10:10, 2013.

\bibitem{Kumar:2011}
Ravi Kumar, Silvio Lattanzi, Sergei Vassilvitskii, and Andrea Vattani.
\newblock Hiring a secretary from a poset.
\newblock In {\em Proceedings of the 12th ACM conference on Electronic
  Commerce}, pages 39--48, 2011.

\bibitem{Lavi:2000}
Ron Lavi and Noam Nisan.
\newblock Competitive analysis of incentive compatible on-line auctions.
\newblock In {\em Proceedings of the 2nd ACM conference on Electronic
  Commerce}, pages 233--241, 2000.

\bibitem{Lindley:1961}
Denis~V Lindley.
\newblock Dynamic programming and decision theory.
\newblock {\em Applied Statistics 10}, pages 39--51, 1961.

\bibitem{Soto:2013}
Jos{\'e}~A Soto.
\newblock Matroid secretary problem in the random-assignment model.
\newblock {\em SIAM Journal on Computing}, 42(1):178--211, 2013.

\end{thebibliography}
\end{document}